\documentclass[11pt]{article}
\usepackage[margin=1in]{geometry}
\usepackage{enumitem}
\usepackage[hyphens]{url}
\usepackage{xcolor}
\usepackage{algpseudocode}
\usepackage{algorithm}
\usepackage{algcompatible}
\usepackage{booktabs} 
\usepackage{graphicx} 
\usepackage{subcaption}
\usepackage{amsmath}
\usepackage{amsthm}
\usepackage{amsfonts}
\usepackage{microtype}
\usepackage{amssymb,bbm,mathrsfs}
\usepackage{pgfplots}
\usepackage{tikz}
\usepackage{multirow}
\usetikzlibrary{graphs}
\usepackage{amsthm}
\newtheorem{theorem}{Theorem}

\title{An Efficient Network-aware Direct Search Method\\for Influence Maximization}

\author{
Matteo Bergamaschi\thanks{University of Padua, Padova, Italy.}
\and
Sara Venturini\thanks{Massachusetts Institute of Technology, Cambridge, MA, USA.}
\and
Francesco Tudisco\thanks{The University of Edinburgh, Edinburgh, Scotland, UK.}
\and
Francesco Rinaldi\footnotemark[1]
}

\date{}

\begin{document}

\maketitle

\begin{abstract}
{Influence Maximization (IM) is a pivotal concept in social network analysis, involving the identification of influential nodes within a network to maximize the number of influenced nodes, and has a wide variety of applications that range from viral marketing and information dissemination to public health campaigns. 
IM can be modeled as a combinatorial optimization problem with a black-box objective function, where the goal is to select $B$ seed nodes that maximize the expected influence spread.  

Direct search methods, which do not require gradient information, are well-suited for such problems. Unlike gradient-based approaches, direct search algorithms, in fact, only evaluate the objective function at a suitably chosen set of trial points around the current solution to guide the search process.  However, these methods often suffer from scalability issues due to the high cost of function evaluations, especially when applied to combinatorial problems like IM.
This work, therefore, proposes the Network-aware Direct Search (NaDS) method, an innovative direct search approach that integrates the network structure into its neighborhood formulation and is used to tackle a mixed-integer programming formulation of the IM problem, the so-called General Information Propagation model. 
We tested our method on large-scale networks, comparing it to existing state-of-the-art approaches for the IM problem, including direct search methods and various greedy techniques and heuristics.
The results of the experiments empirically confirm the assumptions underlying NaDS, demonstrating that exploiting the graph structure of the IM problem in the algorithmic framework can significantly improve its computational efficiency in the considered context.}
{Influence Maximization, Network Analysis, Derivative-Free Optimization.}
\end{abstract}


\section{Introduction}
Influence Maximization (IM) is a key concept in social network analysis. 
It entails identifying a group of influential nodes (seeds) in a network, with the goal of maximizing the resulting number of influenced nodes 
\cite{bakshy2012role,centola2010spread,erkol2019systematic,nekovee2007theory}.
This concept has numerous applications, 
one of the most recognized applications of IM is viral marketing~\cite{domingos2001mining}, which involves a company trying to increase the adoption of a new product through the social networks of initial users. IM also forms the basis for numerous other applications, such as network monitoring~\cite{leskovec2007cost}, rumour control~\cite{budak2011limiting,he2012influence}, and social recommendation~\cite{ye2012exploring}.
The two most popular models to simulate how influence spreads through a network are the Independent Cascade (IC) Model and the Linear Threshold (LT) Model~\cite{kempe2003maximizing,shakarian2015independent}. In the IC model, each influenced node has a single chance to influence its neighbors, and if it succeeds, those neighbors are then activated and have a chance to influence their neighbors in subsequent rounds. Indeed, in the LT model, each node is influenced based on the fraction of its neighbors that are influenced, and a node becomes influenced if the weighted sum of its influenced neighbors exceeds a certain threshold.%\\

IM is typically formulated as a combinatorial optimization problem. Given a budget of $B$ (the number of seed nodes to be selected), the objective is to find the set of $B$ nodes that maximizes the expected influence spread.
Here, we use as our information propagation model the General Information Propagation (GIP) model presented in \cite{tian2022unifying}. Such a model has the following features: it considers continuous variables, enables feedback between nodes, and encompasses two well-known variants of IC and LT models as special cases. Furthermore, the authors in \cite{tian2022unifying} formulate the influence maximization problem as a mixed integer nonlinear programming problem with a black-box objective function (i.e., a function whose analytical expression is not available) and propose a direct search approach, called Customized Direct Search (CDS), to solve it. 
However, the method has computational limitations and does not fully exploit the structure of the graph under analysis. See Section~\ref{sec:prob_set} for more details. 

  Direct Search Methods are well-known and well-suited zeroth-order/derivative-free optimization algorithms, and hence do not rely on derivative information to drive the optimization process (see, e.g., \cite{AuHa2017,conn2009introduction,dzahini2024revisiting,larson2019derivative} and references therein for further details on derivative-free and zeroth-order optimization algorithms). Instead, they make use of function evaluations to guide such a process, making them well suited for non-smooth, noisy, or black-box optimization problems where derivative information is unavailable or unreliable. In particular, they sample the objective function at a finite number of points at each iteration and decide which actions
to take next solely based on those function values, either by function value comparison through ranking or based on numerical values and without any explicit or implicit derivative approximation or model building. These methods are, hence, particularly effective in problems where gradient estimation is computationally expensive or not available, such as IM.
%{GIP Model},{prob_form},{section:cds}
%Derivative-free direct search methods are currently considered {\color{blue}among the state-of-the-art approaches} for this type of optimization problem, as the computation of gradient information is impractical. 
Since function evaluations can be time-consuming and resource-intensive, direct search methods, however, lack scalability and are computationally inefficient when dealing with large-scale problems. Therefore, balancing the trade-off between exploration and computational cost is crucial for practical application.

In order to overcome these limitations, in this work, we propose the Network-aware Direct Search (NaDS) method, an innovative direct search approach that exploits the network structure knowledge thus reducing the computational effort to find good solutions. This is possible thanks to a novel neighborhood definition based on the connections in the graph related to the current iterate. We provide extensive numerical tests that empirically demonstrate the better efficiency of our NaDS method, showing that direct search methods can tackle large-scale networks with thousands of nodes, outperforming baseline approaches in terms of influence score and computational time.

The remainder of the paper is organized as follows.
In Section~\ref{sec:rel_works}, we discuss the existing methods for the IM problem. In Section~\ref{sec:prob_set}, we present the problem, focusing on the mathematical formulation, and introduce the CDS method, which is the main direct search method we use for comparison. In Section~\ref{sec:NaDS}, we present our novel direct search method NaDS. Finally, in Section~\ref{sec:results}, we discuss numerical experiments, and in Section~\ref{sec:concl}, we draw some conclusions and discuss possible future work.

\section{Related Work}
\label{sec:rel_works}
The IM problem gained significant attention due to its applications in marketing, information dissemination, and the spread of health behaviors~\cite{chen2010scalable,bovet2019influence,leskovec2007dynamics,mossel2010submodularity,pastor2015epidemic}.
Due to its large number of applications and complex optimization formulation, it presents many research challenges. On the theoretical modelling side, there are different ways to define a diffusion model that captures information diffusion. On the algorithm side, solving IM optimally is proven to be NP-hard in most diffusion models. Its complexity extends to the evaluation of influence spread from any given seed set, which is also computationally intensive. Consequently, much of the existing research in this area is dedicated to finding approximate solutions.
%Over the past decade, various models and methods have been proposed to address this problem, and surveys concentrating on different aspects are available.
%leading to a wide range of variants.
Over the past decade, a variety of models and methods have been proposed to tackle this problem, and there are available comprehensive surveys that explore different facets of these advancements~\cite{im_survey,guille2013information,sun2011survey,zhang2014recent,chen2022information,arora2017debunking,tejaswi2016diffusion}.
In the following, using the taxonomy proposed in~\cite{im_survey}, we attempt to review and summarize some of these models and algorithms.

%Over the past decade, various models and methods have been proposed to address this challenge, leading to a wide range of variants.
Different models employ distinct mechanisms for capturing how a user transitions from an inactive to an active status due to their neighbors' influence.
Diffusion models can be broadly categorized into progressive and non-progressive diffusion models, depending on whether activated nodes can be deactivated in subsequent steps or not, respectively. Examples of typical non-progressive models include the SIR/SIS model~\cite{kermack1927contribution} and the Voter model~\cite{clifford1973model}.
We will focus on the first class of progressive models and we will briefly present the three most popular models: the Independent Cascade (IC) model, the Linear Threshold (LT) model, and the Triggering (TR) model.

All these three models are diffusion models based on stochastic processes where information propagates from one node to its neighbors at each time step according to different probabilistic rules.
In the IC model~\cite{goldenberg2001talk}, each activated node has a fixed probability of influencing its neighbors in each time step. Once a node is activated, it cannot be deactivated, and each neighbor’s activation is independent of other neighbors’ activations.
In the LT model~\cite{granovetter1978threshold,schelling2006micromotives}, each node has a threshold that must be exceeded by the sum of the influence from its neighbors to become active. Each node is activated based on the weighted sum of its neighbors’ states compared to its own threshold.
In the TR model~\cite{kempe2003maximizing}, each node independently selects a random ``triggering set'' from its neighbors based on a specified distribution, and a node becomes active if it has a neighbor in its chosen triggering set that is active. 
%The TR model generalizes the aforementioned IC and LT models.

These foundational models have paved the way for numerous extensions and adaptations in the literature. For instance, IC, LT, and TR are time-unaware models where diffusion stops when no more nodes can activate. However, real-world propagation campaigns often require maximizing influence within time constraints. This has led to the development of Time-Aware Diffusion Models, categorized as discrete-time~\cite{liu2012time,chen2012time,kim2014ct} and continuous-time models~\cite{rodriguez2011uncovering,xie2015dynadiffuse}.
%In particular, in this work, we refer to \cite{tian2022unifying} which introduces the GIP model, a general class of information propagation model that can explain a wider range of dynamics in the context of information spread (see Section~\ref{GIP Model}). %, leading to a more general, yet complicated, version of the IM problem.
In particular, in this work, we refer to \cite{tian2022unifying}, which introduces the GIP model, which extends the IC and LT to incorporate continuous variables and dynamic feedback mechanisms (see Section~\ref{GIP Model}).
%Two classical propagation models that have been widely studied are the Independent Cascade (IC) model and the Linear Threshold (LT) model \cite{shakarian2015independent}. These foundational models have paved the way for numerous extensions and adaptations in the literature. In particular, in this work we refer to \cite{tian2022unifying} which introduces the GIP model, a general class of information propagation model that can explain a wider range of dynamics in the context of information spread, leading to a more general, yet complicated, version of the IM problem.

State-of-the-art methods for IM include greedy algorithms, which have been shown to be effective since the seminal work of Kempe et al.~\cite{kempe2003maximizing}. These methods iteratively select nodes that provide the largest marginal gain in influence spread. 
They can be categorized into three primary categories depending on their goals:
simulation-based approaches involve  Monte-Carlo simulations to evaluate the influence spread and gain model generality~\cite{kempe2003maximizing}; proxy-based approaches utilize proxy models to approximate the influence function obtaining practical efficiency~\cite{chen2010scalable,chen2009efficient,kimura2006tractable,jung2012irie,liu2014influence}; sketch-based approaches focus on theoretical efficiency by initially constructing theoretically grounded sketches under a diffusion model, which are subsequently used to expedite the evaluation of the influence function~\cite{cheng2013staticgreedy,ohsaka2014fast,cohen2014sketch,borgs2014maximizing}.
%In addition to greedy algorithms, discount heuristics, such as the Degree Discount heuristic proposed by Chen et al.  \cite{chen2010scalable}, have been developed to offer more efficient solutions with competitive performance.
% For readers interested in a more comprehensive review of methods and advances in IM, we refer to \cite{im_survey}, which provides an extensive survey of the topic.
 In addition, methods based on percolation processes or on the concept of k-core centrality have also been proposed as alternatives for identifying influential nodes~\cite{Morone_2015, Kitsak_2010,divideandconquer, 10.1093/comnet/cnz029}.

\section{Problem Setting}
\label{sec:prob_set}%%Influence maximization / theory ecc ecc
We define our social network starting from a graph, denoted as $\mathcal{G}(\mathcal{N}, \mathcal{E}, W)$, which is undirected, connected, and weighted. Here, $\mathcal{N} = \{1, 2, \ldots, n\}$ represents the set of nodes, and $\mathcal{E} = \{(i, j) \mid \text{there exists a connection from node } i \text{ to node } j\}$ represents the set of edges. 

The weight matrix $W = [W_{ij}]$ encodes the strength or level of trust between nodes, where each weight satisfies $W_{ij} \geq 0$, $W_{ij} \in \mathbb{R}$. 
Specifically, $W_{ij} > 0$ if $(i, j) \in \mathcal{E}$, and $W_{ij} = 0$ otherwise, indicating the absence of a direct connection.
%Each edge $(n_i, n_j)$ acts as a conduit for information flow between nodes, with a corresponding weight $W_{ij} > 0$ representing the strength or level of trust between these nodes. When an edge $(n_i, n_j)$ does not exist in $E$, its weight is set to $0$. 
In this context, we employ the notation $x_i(t)$ to denote the state of node $i$ at discrete time step $t$.%\\

Following~\cite{tian2022unifying}, in this section we formalize the influence maximization (IM) problem as well as the associated GIP model, and we review the CDS method which we will then use as a baseline for our improved NaDS scheme. %, following the lines in ~\cite{tian2022unifying}, the IM problem, the GIP model, and the problem formulation are presented.

\subsection{Problem Formulation}
\label{sec:prob_form}
The IM problem revolves around the objective of maximizing the collective impact on network nodes at the culmination of the propagation process. Our particular focus lies on an essential constraint: the initiation of a finite number of nodes determined by a budget size. This constraint represents the notion of allocating limited resources to influence a wider audience.

Given a defined information propagation process and an associated function $s_j = \sum_{t=0}^\infty(1-\gamma)^t x_j(t)$ gauging the overall influence on each node $j$, the cumulative influence across the network emerges naturally as:
\begin{equation}\label{obj}
   s(x) = \sum_{j=1}^n s_j(x(0)), 
\end{equation}
where $x(0)$ represents the initial state vector. Consequently, the IM problem revolves around maximizing $s(x)$, while ensuring adherence to the constraint of activating a limited number of nodes:
\begin{equation}\label{constraint}
    |\{j : x_j(0) > 0\}| \leq B,
\end{equation}
where $B \in \mathbb{N}$ represents the budget size.

Based on the objective function \eqref{obj} and constraint \eqref{constraint}, the IM problem  can be formulated as a mixed-integer nonlinear program (MINLP):
\begin{equation}\label{minlp}
    \begin{aligned}
    \max_{x,z} \quad & s(x) \\
    \text{s.t.} \quad & x_j \leq h_{j,0} \ z_j, \\
    & x_j \geq l_{j,0} \ z_j, \\
    & \sum_j z_j \leq B, \\
    & x_j \in \mathbb{R}, \quad z_j \in \{0, 1\}, \quad \forall j \in [n],
\end{aligned}
\end{equation}
where $0 < l_{j,0} \leq h_{j,0}$ define a lower and an upper bound on the initial influence level of node $j$. The objective function $s(\cdot)$ measures the aggregate influence across the network, as illustrated in \eqref{obj}. The vector $x$ comprises initial state values, while vector $z$, of the same dimensions, signifies the decision to set positive initial state values ($z_j = 1$) or not ($z_j = 0$), thus enforcing the constraint \eqref{constraint}.

Taking advantage of the non-decreasing nature of $s(x)$, the task of maximizing the objective $s(x)$ with respect to both $x$ and $z$ in the MINLP \eqref{minlp} can be simplified to the task of maximizing $s(x)$ while $x$ and $z$ are set to their highest attainable values. Specifically, we set $x_j = h_{j,0} \ z_j$ and $\sum_j z_j = B$, effectively reducing the problem to the following form with respect to the binary vector $z$:
\begin{equation}\label{final}
    \begin{aligned}
    \max_z \quad & s(\bold{h_0} \odot z)\\
        \text{s.t.} \quad & \sum_j z_j = B,\\
        & z_j \in \{0, 1\}, \quad \forall j,
    \end{aligned}
\end{equation}
where $\bold{h_0} = (h_{j,0})$ and $\odot$ denotes the element-wise Hadamard product. As a result, the domain $\Omega^B$ becomes a natural mesh\footnote{In the context of mesh adaptive direct search methods, the mesh is a discrete subset of the parameter space that is adaptively refined or coarsened throughout the optimization process. It provides a framework for generating candidate points around the current best solution.} for searching at every iteration of the algorithm.
\begin{equation}
    \Omega^B = \{ z \in \{0, 1\}^n : \sum_j z_j = B \}.
\end{equation}
The constraints are inherently embedded within the domain, and they are managed using the extreme barrier approach $s_{\Omega^B}$, where $s_{\Omega^B}(z) = s(\bold{h_0} \odot z)$ if $z \in \Omega^B$ and $-\infty$ otherwise. We define the neighborhood function for binary variables $z$ as:
\begin{equation}\label{neighborhood}
    N(z,d) = \{ y \in \Omega^B : \| y - z \|_1 \leq d \},
\end{equation}
where $d \in \mathbb{N}\setminus\{1\}$, as the $L_1$ distance between $y$ and $z$ can reach a minimum of 2 when $y \neq z$ and $y, z \in \Omega^B$. This minimal distance of 2 is achieved when only one element of positive value is exchanged with an element of value $0$.

Given $d \in \mathbb{N}\setminus\{1\}$, we can now formally define a $d$-local maximum of Problem \eqref{final} as a point $z^*$ satisfying the following condition: 

% Given $d \in \mathbb{N}$, a point $z^* \in \{0, 1\}^n$ is a $d$-local
% maximum of Problem \eqref{im}, if
\begin{equation}
s_{\Omega^B}(z^*) \geq s_{\Omega^B}(z), \quad \forall z \in N(z^*, d).
\end{equation}

\subsection{The GIP Model}
\label{GIP Model}
Here, we review the General Information Propagation (GIP) model proposed in~\cite{tian2022unifying}, which unifies the fundamental mechanisms underlying the two classic propagation models. This model encompasses two primary aspects:
\begin{enumerate}[label=(\roman*)]
    \item Each node $i$ independently attempts to influence its neighbors, proportional to the edge weight and its current state value $x_i(t)$. This characteristic aligns with the principles of the IC model.
    \item The actual impact on each node $j$ stems from the collective behaviour of its entire neighborhood.
    \end{enumerate}

The latter point is achieved by applying a nonlinear transformation to
\begin{equation}
    y_j(t) = \sum_i W_{ij} \, x_i(t - 1),
\end{equation}
enabling us to capture how the cumulative influence attempts from all neighbors translate into a state change for node $j$. 
Intuitively, this equation captures the idea that the current state or output  $y_j(t)$ is influenced by a linear combination of prior inputs, where each input $x_i(t-1)$ contributes according to its associated weight $W_{ij}$.
This concept draws parallels not only with the LT model but also with nonlinear models applied to opinion dynamics. Specifically, at each time step $t > 0$, we introduce a lower bound $l_{j,t}$ representing the threshold required to initiate propagation. This implies that $x_j(t) = 0$ if $y_j(t) < l_{j,t}$. Additionally, an upper bound $h_{j,t}$ accounts for the saturation effect, signifying that $x_j(t) = h_{j,t}$ if $y_j(t) \geq h_{j,t}$ . 

Overall, the GIP model can be described as a bounded discrete  dynamical system:
\begin{equation}\label{gip}
    x_j(t) = f_{j,t}\left[\sum_i W_{ij} x_i(t - 1)\right], \quad \text{for all } t > 0, \ j \in V,
\end{equation}
where $f_{j,t}(x)$ takes the form:
\begin{equation}
    f_{j,t}(x) = \begin{cases}
0, & x < l_{j,t}, \\
x, & l_{j,t} \leq x < h_{j,t}, \\
h_{j,t}, & x \geq h_{j,t},
\end{cases}
\end{equation}
representing the time-dependent bounds for each node $j$. %The matrix $W = (W_{ij})$ with $W_{ij} \geq 0$ is the weighted adjacency matrix of the network, and 
$\{l_{j,t}\}$ and $\{h_{j,t}\}$ denote the time-dependent lower and upper bounds for each node $j$ respectively, where $0 \leq l_{j,t} \leq h_{j,t}$. These bound values offer flexibility in characterizing the underlying population. Furthermore, it is worth noting that the GIP model can replicate the classic models by setting specific bound values. The initial states $x(0)$ are predefined, with $x_j(0) \in \{0\} \cup [l_{j,0}, h_{j,0}]$ and $l_{j,0} > 0$.

Specifically we will use the following threshold-type bounds for $t > 0$ and $j \in V$:
\begin{equation}\label{treshold-bound}
   l_{j,t} = (\theta_{l,j} \alpha)^t l_{j,0}, \\ \quad
   h_{j,t} = \theta_{h,j} \, \theta_{l,j}^{t - 1} \alpha^t h_{j,0},
\end{equation}
where $\alpha = \frac{1}{|E|} \sum_{(i,j) \in E} W_{ij}$ is the average across edge weights, while $\theta_{l,j}$ and $\theta_{h,j}$ are the two parameters that characterize the propagation dynamics in the GIP model, as they decide the trend for the lower and the upper bound.
In order to guarantee convergence of the propagation algorithm and ensure that the bounds tend to zero, it is required that $\theta_{l,j} \alpha < 1$.

For a detailed explanation of the propagation process, we refer to Algorithm~\ref{propalgo}.

\begin{algorithm}[t]
\caption{Propagation Algorithm for the GIP Model}
\begin{algorithmic}[1]
\REQUIRE A undirected, connected, and weighted network $\mathcal{G}(\mathcal{N}, \mathcal{E}, W)$ with non-negative weights,
%matrix $W$ where $W_{ij} > 0$ if $(n_i, n_j) \in E$
parameters $\{l_{j,t}, h_{j,t}\}$ in the GIP model, time-discounting factor $\gamma$, initial state $x(0)$ where $x_j(0) \in [l_{j,0}, h_{j,0}]$ if and only if $j \in A_0$ ($0$ otherwise), and the tolerance $\epsilon$.
\ENSURE The value of the objective function in the MINLP, $s$.
\STATE Set $t \gets 0$, $x(t) \gets x(0)$, and $s \gets 0$.
% \STATE Mark all out-neighbors of $A_0$ as potentially activated nodes, $N_t \gets \bigcup_{j \in A_0} N^{\text{out}}(j)$.
\STATE Set $N_t \gets \bigcup_{j \in A_0} N^{\text{out}}(j)$, where $N^{\text{out}}(j) := \{ k \in V \mid (j, k) \in E \}$.
\WHILE{$|(1 - \gamma)^t x(t)| > \epsilon$}
    \STATE $A_{t+1}, N_{t+1} \gets \varnothing$, and $x(t+1) \gets 0$.
    \FOR{each potentially activated node $j \in N_t$}
        \STATE $x(t+1)_j \gets f_{j,t} \left( \sum_{i \in A_t} W_{ij} x(t)_i \right)$.
        \IF{$x(t+1)_j > 0$}
            \STATE $A_{t+1} \gets A_{t+1} \cup \{j\}$.
            \STATE $N_{t+1} \gets N_{t+1} \cup N^{\text{out}}(j)$.
            \STATE $s \gets s + (1 - \gamma)^{t+1} x(t+1)_j$.
        \ENDIF
    \ENDFOR
    \STATE $t \gets t + 1$.
\ENDWHILE
\end{algorithmic}
\label{propalgo}
\end{algorithm}

\subsection{CDS Method}
\label{section:cds}

In~\cite{tian2022unifying}, after introducing the GIP Model (see Section~\ref{GIP Model}), the CDS method is proposed to solve it.

The algorithm starts with the node set that maximizes the Katz centrality \cite{Katz_1953}, which represents an exact solution in certain GIP settings. At each iteration $r$, we start with a \textbf{SEARCH step} that suitably explores the points in the domain $\Omega^B$ in order to possibly find a new iterate that improves the objective function value. If this is the case, we directly switch to the \textbf{Termination check}. Otherwise, we perform a \textbf{POLL step}. Such a step explores the local neighborhood of the current candidate $z(r)$ until a point with a sufficient improvement in the objective value is found or all points have been examined. In the termination check, if an improved point is identified, the algorithm will start a new iteration but will decrease the required improvement if a sufficiently improved point has not been discovered. If no improvement is found, the algorithm outputs the current iterate and terminates. We refer to Algorithm~\ref{cds} for a detailed description of steps.

Hence, the termination step directly guarantees convergence to local minima. Even though convergence to global optimal solutions could be achieved through a carefully developed search step and a better understanding of the objective function's landscape to avoid poor local optima, the drawback of global optimization methods is their time complexity. Thus, we maintain the search step as an optional component. The CDS method takes advantage of the problem's characteristics and mitigates worst-case complexity by initializing with an exact solution when the GIP model reaches the extreme of the EIC model. 
% It is conjectured that the local optima near this specific solution are sufficiently good in terms of the objective function value.
As conjectured in \cite{tian2022unifying}, the local optima in the vicinity of this particular solution are expected to provide sufficiently good values of the objective function.

 From the current CDS method, two avenues for enhancing output quality emerge:
%The IM problem is equivalent to selecting a set of nodes with values $1$ (and others with $0$). Hence, two known methods for guaranteeing convergence to global optima can be applied
 (i) brute force, where all node sets of size $k$ are evaluated to select the optimal one, and (ii) random sampling, which exhibits asymptotic global convergence if it samples densely enough. Two improvement strategies inspired by those methods can hence be used in our context. Firstly, by expanding the neighborhood definition, more points are searched within the feasible domain. If the neighborhood encompasses the entire domain, we get back to the brute-force method. Secondly, the search process (steps 2, 3, and 4 in Algorithm~\ref{cds}) can be restarted from other randomly generated unexplored points, mirroring the logic of the search step. This strategy achieves asymptotic convergence to global minima, akin to the random sampling approach.

\begin{algorithm}[t]
\caption{Customized Direct Search (CDS) Method}
\begin{algorithmic}[1]
\STATE \textbf{Initialization:} Set $0 < \zeta, \delta < 1, d > 1$. Let $z(0) \in \Omega^B$ such that $z(0)_j = 1$ if node $j \in A = \{j_1, \ldots, j_k\}$ where $h_{i,0} \, c_i = h_{j,0} \, c_j, \forall i \notin A, j \in A$, and $c = \{[I - (1 - \gamma)W]^{-1} - I\} \mathbf{1}$ is the Katz centrality. Set iteration $r = 0$.
\STATE \textbf{SEARCH step (optional):} Evaluate $s_{\Omega^B}$ on a finite subset of trial points on the domain $\Omega^B$, until a sufficiently improved point $z$ is found, where $s_{\Omega^B}(z) > (1 + \zeta)s_{\Omega^B}(z^{(r)})$, or all points have been exhausted. If an improved point is found, then the SEARCH step may terminate; skip the next POLL step and go directly to step 4.
\STATE \textbf{POLL step:} Evaluate $s_{\Omega^B}$ on the set $N(z^{(r)}, d) \subset \Omega^B$ as in \eqref{neighborhood}, until a sufficiently improved point $z$ is found, where $s_{\Omega^B}(z) > (1 + \zeta)s_{\Omega^B}(z^{(r)})$, or all points have been exhausted.
\STATE \textbf{Termination check:} If an improvement is found, set $z^{(r+1)}$ as the improved solution, while decreasing $\zeta \leftarrow \delta \zeta$ if a sufficient improvement has not been found, increment $r \leftarrow r + 1$, and go to step 2. Otherwise, output the solution $z^{(r)}$.
\end{algorithmic}
\label{cds}
\end{algorithm}

\section{NaDS Method}
\label{sec:NaDS}
In this section, we present the NaDS method, which  takes advantage of the inherent structure of the network to reduce the computational cost of the exploration strategy, thus speeding up the search process in CDS.

In the context of direct search methods, a central challenge revolves around the selection of an appropriate  neighborhood, a crucial factor that significantly influences the way solution space is explored. In the preceding section, we reviewed the CDS method, which employs a rather straightforward search strategy. Specifically, it populates the neighborhood with all points located at a specified $L_1$ distance, commonly set at $2$. Moreover, the CDS method incorporates the graph structure by utilizing Katz centrality to guide the initial point selection. However, these measures do not adequately take into account the underlying graph structure. In this section, we introduce a variation of CDS that directly exploits the network structure in the neighborhood construction phase. We call the resulting algorithm Network-aware Direct Search (NaDS) method.

The proposed NaDS method deviates primarily from CDS in its trial point selection. This new methodology aims to better leverage the underlying graph structure during neighborhood construction. The core notion is centered around a more conservative selection of trial points, emphasizing connections within the graph.
The fundamental premise of the NaDS method is to exclusively consider points that are interconnected through graph edges. Unlike the CDS inclusive approach of exploring all points at a set $L_1$ distance, the NaDS method narrows down its focus to only those points that are directly linked through edges in the graph.
More specifically, in the main poll step, NaDS filters the trial points through the set $C$, defined as follows:
\begin{equation}
\label{czetadef}
    C(z) = \{ y \in \Omega^B |\ y_i = 1 \Rightarrow \exists\ j\ \mbox{s.t.}\ (i,j) \in E \lor z_i = 1 \}.
\end{equation}
Here, $C(z)$ represents the set of points that can be constructed using only nodes that are either:
\begin{enumerate}[label=(\roman*)]
    \item already active in $z$ (i.e., $z_i = 1$),
    \item connected to an active node in $z$ (i.e., $\exists$ $j$ such that $(i,j) \in E$).
\end{enumerate}
Figure~\ref{fig:mesh_comparison} provides an example that illustrates the differences in neighborhood construction between CDS and NaDS.

\begin{figure}[t]
    \centering
    \includegraphics[width=0.3\columnwidth]{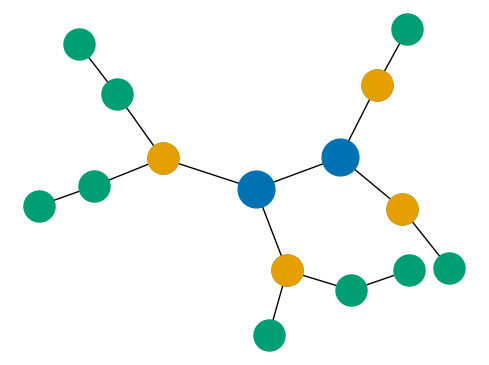}
    \caption{Simple example of how nodes are selected in CDS and NaDS. Given a set of activated nodes (in blue), CDS selects potential swap nodes from all the remaining nodes of the graph (in green and orange), while NaDS only considers the ones connected to the starting nodes (in orange).}
    \label{fig:mesh_comparison}
\end{figure}

%In summary, the NaDS method seeks to enhance the mesh construction process by incorporating the graph's topological structure more effectively. Unlike the CDS method, which indiscriminately covers a designated $L1$ distance, the NaDS method indeed opts for a more graph-aware approach by exclusively considering neighboring points connected through edges. This approach aims to provide a refined exploration of the solution space, accounting for the inherent relationships within the underlying graph structure.
The NaDS method operates within a more constrained neighborhood, leading to a reduction in the number of function evaluations required at each iteration. Importantly, this reduction in neighborhood size does not compromise the quality of the search. This is due to the fact that the NaDS method hones in solely on the most significant and meaningful points within the solution space.

This approach ensures a more focused exploration by considering points that are directly linked in the graph rather than exhaustively covering a broader $L_1$ distance as done by the CDS method. By deliberately selecting only interconnected points, the NaDS method strikes a balance between the efficiency and depth of the search.

In Algorithm~\ref{NaDS}, we present the detailed scheme of the NaDS method. In general, we do not start from a solution based on Katz centrality like CDS.
% Indeed, we observed during our tests that there are better and faster heuristics to use as starting points
In fact, our experiments showed that more effective and faster heuristics can be used as starting points, depending on the graph's characteristics and size — for example, Single Discount and Simple Greedy — as demonstrated in Section \ref{sec:results}. Then, at each iteration $r$, we start with an optional \textbf{SEARCH step} that suitably explores the points in the mesh $\Omega^B$ in order to possibly find a new iterate that improves the objective function value. If this is the case, we directly switch to the \textbf{Termination check}; otherwise, we perform the \textbf{POLL step}. 

In the first phase of the poll step, we explore the refined neighborhood defined in \eqref{czetadef}. This is the core of the poll step, and as our tests empirically confirm, almost every improving point is found in this phase. In case this phase cannot find an improving point, NaDS performs \textbf{two additional POLL step phases}, which also ensure local convergence, as in the case of the CDS method. In the second phase, the refined neighborhood \eqref{czetadef} is dynamically expanded to match the $L_1$ neighborhood defined in \eqref{neighborhood}. We remark that in this phase, we set the distance $d$ to the minimum, which is 2. The third phase consists of expanding the aforementioned neighborhood by increasing the distance $d$. In general, this third phase is optional due to the high computational cost it requires, but it is probably useful in smaller networks.

% The construction of the NaDS method, and in particular the second phase of the poll step, guarantees the local convergence of the method. In particular, we can state the following results on the nature of NaDS solutions.
%
The local convergence of the NaDS method is guaranteed by construction since it operates within a specific neighborhood and will, therefore, find a local minimum within that area. Precisely, we have
\begin{theorem}
Let $\{z_k\}$ and $\{s_k\}$ be the sequences of solutions and reference values generated by NaDS. Then, the algorithm produces a $d$-local maximum for any fixed $d$ in finite time. 
\end{theorem}

\begin{proof}
Let $\{z_k\}$ and $\{s_k\}$ be the sequences of solutions and reference values generated by NaDS.  
Since $z_k$ is the output of NaDS, the \textbf{POLL step} phase ensures that every point $z_t$ in the neighborhood $N(z_k, d)$ has been evaluated, and there is no point such that $s_t > s_k$. This guarantees that $z_k$ is a $d$-local maximum. To prove that NaDS produces the sequence $\{z_k\}$ in finite time, we observe that $s_0 < s_1 < \dots < s_k$, which implies that NaDS does not visit previously evaluated points twice.  
Furthermore, since the bounded domain $\Omega^B$ contains a finite number of points, it follows that the sequence $\{z_k\}$ must also be finite.

\end{proof}

\begin{algorithm}[t]
\caption{Network-aware Direct Search (NaDS) Method}
\begin{algorithmic}[1]
\STATE \textbf{Initialization:} Set $0 < \zeta, \delta < 1, d = 2,  d_{max} > 1$. Let $z(0) \in \Omega^B$ be a suitably chosen starting point. Set iteration $r = 0$.

\STATE \textbf{SEARCH step (optional):} Evaluate $s_{\Omega^B}$ on a finite subset of trial points on the domain $\Omega^B$, until a sufficiently improved point $z$ is found, where $s_{\Omega^B}(z) > (1 + \zeta)s_{\Omega^B}(z^{(r)})$, or all points have been exhausted. If an improved point is found, then the SEARCH step may terminate, skip the next POLL step and go directly to step 6.

\STATE \textbf{POLL step (Phase 1):} Evaluate $s_{\Omega^B}$ on the set $N(z^{(r)},d)\cap C(z^{(r)}) \subset \Omega^B$ as in \eqref{neighborhood} with distance $d$, until a sufficiently improved point $z$ is found, where $s(\Omega^B)(z) > (1 + \zeta)s(\Omega^B)(z^{(r)})$, or all points have been exhausted. $C(z^{(r)})$ is the set of points that are connected by an arc with $z^{(r)}$. If an improved point is found, then the POLL step may
terminate, and go directly to step 6.
\STATE \textbf{POLL step (Phase 2):} Expand dynamically $N(z^{(r)},d)\cap C(z^{(r)})$ to $N(z^{(r)},d)$. If an improved point is found, then the POLL step may
terminate, skip Phase 3, and go directly to step 6.
\STATE \textbf{POLL step (Phase 3 - Optional):} If $d < d_{max}$ then increase $d$. 
\STATE \textbf{Termination check:} If an improvement is found, set $z^{(r+1)}$ as the improved solution, while decreasing $\zeta \leftarrow \delta \zeta$ if a sufficient improvement has not been found, increment $r \leftarrow r + 1$, and go to step 2. Otherwise, output the solution $z^{(r)}$.
\end{algorithmic}
\label{NaDS}
\end{algorithm}

\section{Numerical Experiments}\label{sec:results}

In this section, we evaluate the NaDS method through extensive experiments, showing its superiority over existing approaches. Our results indeed show that NaDS consistently outperforms both the CDS method introduced in \cite{tian2022unifying} and classic heuristics across all tested networks, achieving higher influence scores in 23 over 24 test cases. The network-aware neighborhood construction is particularly effective, reducing the search space by 35-60\%, while maintaining solution quality. The method shows robust performance across diverse network structures and budget constraints, from small collaboration networks to large social graphs. The following sections detail our experimental setup on six real-world networks (Table~\ref{tab:graph_summary}) with budgets $B \in \{5, 10, 15, 20\}$, comparing against five baseline methods while evaluating both influence spread and computational efficiency. We also conducted  experiments on artificial graphs; see the Appendix (Section~\ref{sec:artificial}) for a detailed overview.

% In this section, we conduct a series of experiments on real-world networks. Our main goal is to showcase the enhanced efficiency of our improved direct search method, NaDS, in comparison to the CDS method introduced in \cite{tian2022unifying} for addressing the IM problem. 

The complete code for the implementation of methods and experiments is available at \url{https://github.com/Berga53/Network-aware-Direct-Search}.
We acknowledge the use of the open-source libraries NetworkX \cite{SciPyProceedings_11}, NumPy \cite{harris2020array}, SciPy \cite{virtanen2020scipy}, and Matplotlib \cite{4160265} in the development of our methods and experiments.

Given that greedy algorithms and heuristics are widely used for this problem, we also incorporate several of these methods into our experiments. The objectives of these experiments are outlined below:
%In this section, we present a series of experiments conducted on real-world networks. The primary aim of these experiments is to demonstrate the superior efficiency of our improved direct search method, NaDS, compared to the CDS method from \cite{tian2022unifying} for the IM problem. Given that the most commonly used methods for this problem are greedy algorithms and heuristics, we will also include several of these methods in our experiments. The objectives of these experiments are as follows:
\begin{enumerate}[label=(\roman*)]
    \item to compare the performance of NaDS against the CDS method and show that NaDS achieves better results;
    \item to illustrate that greedy methods are not suitable for every IM setting, and specifically, demonstrate that these methods lack guarantees and yield worse results in our particular setting;
    \item to show that NaDS provides better and more reliable results compared to the best heuristics available in the literature for the IM problem.
\end{enumerate}

\begin{table}[t]
\centering
%\textmd
\caption{Basic statistics for the real-world datasets, including the reference, number of nodes, and number of edges for each dataset.}
\begin{tabular}{llrr}
\toprule
Dataset      & Ref. & \#nodes       & \#edges       \\
\midrule
Arxiv Astro & \cite{10.1145/1217299.1217301} & 18772 & 198110 \\
Arxiv Gr-Qc & \cite{10.1145/1217299.1217301} & 5242 & 14496  \\
Arxiv Hep-Ph & \cite{10.1145/1217299.1217301} & 12008 & 118521 \\
Lastfm-Asia & \cite{feather} & 7624  & 27806  \\
email-Enron & \cite{leskovec2008communitystructurelargenetworks}\cite{enron} & 36692 & 183831 \\
Facebook & \cite{10.5555/2999134.2999195}    & 4039  & 88234 \\
\bottomrule
\end{tabular}
\label{tab:graph_summary}
\end{table}

The six networks used for our experiments details are reported in Table~\ref{tab:graph_summary}. 
% It is important to note that these are large networks, especially considering the limited computational resources available during the experiments.
%\input{Tables/graph summary}
We selected networks from various applications of the IM problem, which include the science of science (Arxiv Astro, Arxiv Gr-Qc, Arxiv Hep-Ph networks) and opinion dynamics (Lastfm-Asia, email-Enron, Facebook).
All the datasets were selected from the SNAP dataset \cite{snapnets}. Below we give a brief description:
\begin{enumerate}
    \item \textbf{Arxiv Astro} \cite{10.1145/1217299.1217301}: collaboration network from the e-print arXiv restricted to 
    papers in the Astro Physics category. Nodes represent authors, and edges indicate collaborations between pairs of authors.
    The data covers papers in the period from January 1993 to April 2003. 
    \item \textbf{Arxiv Gr-Qc} \cite{10.1145/1217299.1217301}: collaboration network from the e-print arXiv restricted to 
    papers in the General Relativity and Quantum Cosmology category. Nodes represent authors, and edges indicate collaborations between pairs of authors.
    The data covers papers in the period from January 1993 to April 2003. 
    \item \textbf{Arxiv Hep-Ph} \cite{10.1145/1217299.1217301}: collaboration network from the e-print arXiv restricted to 
    papers in the High Energy Physics - Phenomenology category. Nodes represent authors, and edges indicate collaborations between pairs of authors.
    The data covers papers in the period from January 1993 to April 2003.
    \item \textbf{Lastfm-Asia} \cite{feather}: social network from the LastFM music website.
    Nodes are LastFM users from Asian countries and edges are mutual follower relationships between them.
    The dataset was collected from the public API in March 2020. 
    \item \textbf{email-Enron} \cite{leskovec2008communitystructurelargenetworks}\cite{enron}: 
    %communication network that covers all the email communication within a dataset of around half a million emails.
    communication network containing around half a million emails generated by employees of the Enron Corporation. 
    This data was originally made public, and posted to the web, by the Federal Energy Regulatory Commission during its investigation. Nodes represent email addresses, and edges indicate email communication between pairs of addresses.
    %the american energy, commodities, and services company Enron
    \item \textbf{Facebook} \cite{10.5555/2999134.2999195}: social network  
    consisting of "circles" (or "friends lists") from Facebook. The dataset was collected from survey participants using the Facebook app. Nodes represent Facebook users, and edges are mutual follower relationships between them.
\end{enumerate}
The methods considered in our experiments are:
\begin{itemize}
\item \textbf{NaDS Method}: the direct search method proposed in this paper;
    \item \textbf{CDS Method} \cite{tian2022unifying}: the direct search method presented in Section~\ref{section:cds};
    %The direct search method currently existing in the literature.
    \item \textbf{Other Non-Direct Search Methods}:
    \begin{itemize}
        \item \textbf{Single Discount (SD)} \cite{chen2010scalable}: a heuristic that selects nodes based on their centrality, adding a penalty for those that are connected to other high-centrality nodes.
        \item \textbf{Simple Greedy (SG)} \cite{kempe2003maximizing}: a classic greedy algorithm for IM, which iteratively selects nodes to maximize the spread of influence. At each step, it evaluates the marginal gain in influence spread for each candidate node and adds the node with the highest gain to the seed set.
        \item \textbf{Katz Centrality (KC)} \cite{tian2022unifying}: a heuristic that selects the nodes with the highest value of Katz centrality score, a centrality measure that accounts for the number of paths from a node to all others, weighted by path length.
        \item \textbf{K-core Centrality (CC)} \cite{Kitsak_2010}: a heuristic that selects nodes by their core number, defined as the largest k-core in which the node is a member. A k-core refers to a maximal subgraph where each node has a degree that is equal to or greater than k.
        \item \textbf{Collective Influence (CI)} \cite{Morone_2015}: a heuristic that selects nodes with the highest value of Collective Influence, a centrality measure designed to identify nodes whose removal most efficiently fragments the network. CI accounts for both a node’s local degree and the degrees of nodes at the boundary of a ball of fixed radius.
    \end{itemize}
\end{itemize}    

\begin{figure}[t]
    \centering
    \includegraphics[width=0.8\columnwidth]{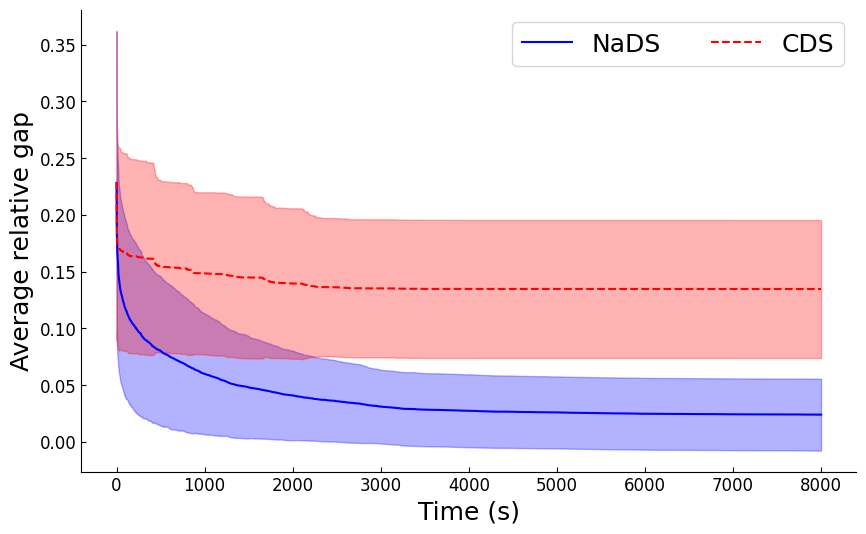}
    \caption{
   Results of the second set of experiments where we compare the average relative gap over time 
 (with shaded standard deviation) of the NaDS and CDS methods, resp. in blue and red.}
    \label{fig:avggap}
\end{figure}

\begin{table}[t]
 
\caption{Performance comparison based on cumulative influence from the first set of experiments, evaluated across different datasets and budget levels. The best results are highlighted in bold. NaDS (mean) is computed accross 10 different realizations.}
\centering
\footnotesize
\setlength{\tabcolsep}{1.2pt}
\begin{tabular}{lcccccccccccc} 
\toprule
\multirow{2}{*}{\textbf{Dataset}} & \multirow{2}{*}{\textbf{B}} 
& \multicolumn{3}{c}{\textbf{NaDS}} 
& \multicolumn{3}{c}{\textbf{CDS}} 
& \multirow{2}{*}{\textbf{SD}} 
& \multirow{2}{*}{\textbf{SG}} 
& \multirow{2}{*}{\textbf{KC}} 
& \multirow{2}{*}{\textbf{CC}}
& \multirow{2}{*}{\textbf{CI}} \\
\cmidrule(lr){3-6} \cmidrule(lr){6-8}
& & Mean & Best & Disc & Mean & Best & Disc & & & & & \\
\midrule
\multirow{4}{*}{Arxiv Astro}
 & 5  & 3751.43 & \textbf{3815.8} & \textbf{3815.8} & 3373.38 & 3749.06 & 3597.15 & 3571.59 & 2049.12 & 357.35 & 1934.67 & 3186.73 \\
 & 10 & 5068.22 & \textbf{5083.1} & \textbf{5083.1} & 4404.69 & 4592.02 & 4989.82 & 4853.77 & 4031.24 & 459.94& 2602.13 & 4560.46  \\
 & 15 & 5986.41 & \textbf{5989.09} & 5977.78 & 5221.54 & 5472.44 & 5909.25 & 5796.86 & 5315.98 & 795.46 & 3008.36 & 5501.35  \\
 & 20 & 6660.95 & \textbf{6671.49} & \textbf{6671.49} & 5761.73 & 5972.2 & 6558.47 & 6481.9 & 6304.03 & 1149.91 & 3325.22 & 6284.10  \\

\midrule
\multirow{4}{*}{Arxiv Gr-Qc}
 & 5  & 199.31 & 199.97 & 198.74 & 203.1 & \textbf{220.22} & 198.74 & 186.45 & 4.29 & 34.18 & 184.29 & 172.78  \\
 & 10  & 280.86 & \textbf{307.13} & 277.39 & 271.9 & 292.97 & 276.94 & 275.9 & 35.53 & 74.69 & 245.79 & 237.71 \\
 & 15  & 375.32 & \textbf{379.4} & 340.92 & 341.73 & 350.05 & \textbf{379.4} & 326.85 & 188.29 & 100.49 & 307.29 & 295.79 \\
 & 20   & 446.26 & \textbf{460.46} & 396.38 & 410.38 & 430.1 & 438.01 & 382.93 & 270.98 & 121.66 & 359.82 & 340.20 \\

\midrule
\multirow{4}{*}{Arxiv HeP-Ph}
 & 5   & 2329.97 & \textbf{2339.73} & 2339.73 & 2263.14 & 2314.29 & 2317.71 & 2293.24 & 2109.6 & 85.91 & 2300.30 & 2020.11 \\
 & 10  & 3008.04 & \textbf{3043.62} & 3000.11 & 2851.42 & 2899.88 & 2989.11 & 2965.51 & 2822.33 & 456.01 & 2965.91 & 2532.34 \\
 & 15  & 3519.25 & \textbf{3533.07} & 3500.4 & 3275.9 & 3340.8 & 3490.11 & 3479.19 & 3380.15 & 628.78 & 3465.91 & 2915.90 \\
 & 20  & 3939.2 & \textbf{3946.2} & 3936.05 & 3657.36 & 3722.2 & 3934.66 & 3910.11 & 3842.82 & 695.57 & 3822.11 & 3282.54 \\

\midrule
\multirow{4}{*}{Lastfm-Asia}
 & 5  & 452.68 & \textbf{487.06} & 406.41 & 351.21 & 449.76 & 407.8 & 188.00 & 178.70 & 126.54 & 409.80 & 99.40  \\
 & 10  & 773.71 & \textbf{815.85} & 806.61 & 695.56 & 735.28 & 697.32 & 571.75 & 318.26 & 415.93 & 534.06 & 557.20 \\
 & 15 & 1036.49 & \textbf{1077.62} & 1090.26 & 866.06 & 944.14 & 976.72 & 913.53 & 410.2 & 587.65 & 662.31 & 873.06 \\
 & 20  & 1280.09 & \textbf{1297.63} & 1280.41 & 1003.63 & 1085.63 & 1207.27 & 1158.69 & 618.8 & 657.01 & 747.04 & 1065.35 \\

\midrule
\multirow{4}{*}{email-Enron}
 & 5  & 7679.31 & \textbf{7689.81} & \textbf{7689.81} & 7001.57 & 7505.15 & 6435.79 & 5960.03 & 879.02 & 4183.08 & 5138.12 & 4862.94 \\
 & 10  & 9955.37 & \textbf{10018.82} & 9995.02 & 8824.24 & 9153.67 & 9480.59 & 9359.18 & 8066.82 & 5276.44 & 7221.76 & 7355.98 \\
 & 15  & 11467.5 & \textbf{11589.03} & 11588.8 & 9993.2 & 10482.51 & 11197.63 & 11189.28 & 10162.89 & 6219.94 & 8267.10 & 9143.80 \\
 & 20  & 12738.43 & \textbf{12925.81} & 12841.32 & 11056.62 & 11500.55 & 12429.74 & 12429.74 & 11869.34 & 7025.14 & 9262.36 & 10843.41 \\

\midrule
\multirow{4}{*}{Facebook}
 & 5  & \textbf{1570.67} & \textbf{1570.67} & \textbf{1570.67} & 1323.08 & \textbf{1570.67} & 1424.51 & 380.37 & 890.93 & 297.08 & 589.48 & 380.37 \\
 & 10  & 2032.22 & \textbf{2037.63} & 2028.97 & 1825.81 & 1974.28 & 1998.38 & 1601.01 & 1537.71 & 447.29 & 781.12 & 1066.63 \\
 & 15  & 2368.61 & \textbf{2639.21} & 2361.72 & 2071.97 & 2472.67 & 2561.29 & 2065.6 & 1800.75 & 555.14 & 954.58 & 1761.60 \\
 & 20  & 2966.17 & 2972.91 & \textbf{2999.34} & 2197.34 & 2588.8 & 2738.75 & 2239.41 & 2413.19 & 1013.8 & 1084.30 & 1934.85 \\
\bottomrule
\end{tabular}
\label{tab:method_comparison}
\end{table}

\begin{table}[t]

\caption{Average relative gaps for NaDS and CDS at different percentages of the time budget. Relative gaps are expressed in \%.}
\centering
\footnotesize
\setlength{\tabcolsep}{2pt}
\begin{tabular}{lccccccccccccc}
\toprule
\multirow{2}{*}{\textbf{Dataset}}  & \multirow{2}{*}{\textbf{\#nodes}}  & \multirow{2}{*}{\textbf{B}} & 
\multirow{2}{*}{\textbf{Time Budget}}
& \multicolumn{5}{c}{\textbf{NaDS}} 
& \multicolumn{5}{c}{\textbf{CDS}} \\
\cmidrule(lr){5-9} \cmidrule(lr){10-14}
& & & & 15\% & 30\% & 50\% & 75\% & 100\% & 15\% & 30\% & 50\% & 75\% & 100\% \\
\midrule
\multirow{4}{*}{Arxiv Astro}   &   \multirow{4}{*}{18772}     &  5 & 2000 & 2.931 & 0.100 & 0     & 0     & 0     & 10.118 & 10.118 & 10.118 & 10.118 & 10.118 \\
              &       & 10  & 4000 & 6.743 & 2.033 & 0.066 & 0     & 0     & 13.708 & 13.708 & 13.708 & 13.708 & 13.088 \\
              &       &  15 & 6000 & 7.385 & 3.249 & 0.759 & 0     & 0     & 13.152 & 13.152 & 13.152 & 12.775 & 12.775 \\
              &       & 20 & 8000 & 8.517 & 3.144 & 0.815 & 0.051 & 0     & 13.667 & 13.667 & 13.542 & 13.542 & 13.499 \\
\midrule
\multirow{4}{*}{Arxiv Gr-Qc }  &   \multirow{4}{*}{52442}    &  5 & 1000 & 2.976 & 2.976 & 2.976 & 2.976 & 2.976 & 3.541  & 3.040  & 2.655  & 1.802  & 1.323  \\
              &       &  10 & 2000 & 2.241 & 1.328 & 1.328 & 1.328 & 1.328 & 6.972  & 6.692  & 6.635  & 4.816  & 4.515  \\
              &       &  15 & 3000 & 0.320 & 0.116 & 0.001 & 0.001 & 0     & 11.625 & 11.084 & 11.072 & 9.087  & 8.954  \\
              &       &  20 & 4000 & 0.296 & 0.207 & 0.128 & 0     & 0     & 10.979 & 10.533 & 10.134 & 8.140  & 8.023  \\
\midrule
\multirow{4}{*}{Arxiv Hep-Ph}  &  \multirow{4}{*}{12008}     & 5 & 1000 & 1.165 & 0.046 & 0     & 0     & 0     & 2.864  & 2.864  & 2.864  & 2.864  & 2.864  \\
              &       &  10 & 2000 & 2.456 & 0.829 & 0.047 & 0     & 0     & 5.445  & 5.445  & 5.445  & 5.208  & 5.208  \\
              &       &  15 & 3000 & 4.415 & 1.990 & 0.402 & 0.013 & 0     & 7.352  & 7.352  & 7.072  & 6.916  & 6.916  \\
              &       &  20 & 4000 & 4.925 & 2.261 & 0.435 & 0.015 & 0     & 7.407  & 7.202  & 7.202  & 7.163  & 7.154  \\
\midrule
\multirow{4}{*}{Lastfm-Asia}   &   \multirow{4}{*}{7624}    & 5 & 1000 & 0.052 & 0.052 & 0.052 & 0.052 & 0.052 & 30.914 & 30.914 & 25.174 & 25.174 & 22.687 \\
              &       & 10  & 2000 & 1.282 & 1.187 & 0.425 & 0.220 & 0     & 21.788 & 16.684 & 14.248 & 12.787 & 9.909  \\
              &       & 15  & 3000 & 1.664 & 0.797 & 0.155 & 0.095 & 0     & 20.747 & 19.650 & 18.520 & 16.949 & 16.478 \\
              &       &  20 & 4000 & 1.663 & 0.738 & 0.120 & 0     & 0     & 24.426 & 23.057 & 22.050 & 21.672 & 21.554 \\
\midrule
\multirow{4}{*}{Email-Enron}     &  \multirow{4}{*}{36692}     &  5 & 2000 & 7.911 & 4.246 & 1.363 & 0.104 & 0     & 8.818  & 8.818  & 8.818  & 8.818  & 8.818  \\
              &       &  10 & 4000 & 10.098 & 6.845 & 3.266 & 1.068 & 0     & 11.378 & 11.378 & 11.378 & 11.378 & 11.378 \\
              &       &  15 & 6000 & 10.072 & 7.993 & 4.459 & 1.715 & 0     & 12.848 & 12.848 & 12.848 & 12.848 & 12.848 \\
              &       & 20  & 8000 & 11.407 & 9.291 & 4.746 & 1.794 & 0     & 13.197 & 13.197 & 13.197 & 13.197 & 13.197 \\
\midrule
\multirow{4}{*}{Facebook}   &   \multirow{4}{*}{4039}    &  5 & 1000 & 8.514 & 2.981 & 1.861 & 0     & 0     & 24.188 & 23.839 & 21.063 & 18.261 & 15.763 \\
              &       &  10 & 2000 & 9.221 & 3.965 & 1.527 & 0.289 & 0     & 14.532 & 13.893 & 13.604 & 10.600 & 10.167 \\
              &       & 15  & 3000 & 17.274 & 10.729 & 5.392 & 2.149 & 0     & 21.401 & 20.723 & 20.718 & 12.832 & 12.515 \\
              &       & 20  & 4000 & 24.546 & 16.406 & 9.475 & 2.119 & 0     & 30.943 & 30.936 & 30.876 & 25.920 & 25.915 \\
\bottomrule
\end{tabular}
\label{tab:relative_gaps_time}
\end{table}

We remark that classic methods and heuristics have a significantly lower computational cost, albeit at the expense of solution quality, as demonstrated by the results presented below.
Before presenting the results, we define the settings for the GIP model used in our experiments. As shown in Equation \eqref{treshold-bound}, the parameters $\theta_l$ and $\theta_h$ fully define the behavior of the GIP model. We have chosen the following values for our experiments:
\begin{itemize}
    \item $\theta_l = 2$, i.e. it requires at least 2 nodes to activate another node; 
    \item $\theta_h = 50$: which ensures variance in node behavior.
\end{itemize}
Using this setting, the first observation is that KC generally performs poorly compared to other heuristics, both in terms of computation time and solution quality as we can see in Table~\ref{tab:method_comparison}. Consequently, we decided to use the SD heuristic as a starting point for both CDS and NaDS. 
%This decision led to a noticeable improvement in the performance of both algorithms.
We then conducted two sets of experiments on the six networks mentioned above, varying the budget across four values: $B = 5, 10, 15, 20$, which is connected with the complexity of the problem.

In the first set of experiments, we used the SD heuristic as the starting point for both methods.
Table~\ref{tab:method_comparison} shows the best results achieved and the method that achieved them, in terms of cumulative influence as defined in \eqref{obj}.
This approach ensures that CDS and NaDS consistently achieve the best solutions across all the methods considered. One of the strengths of the direct search approach is its ability to be combined with heuristics to achieve better results, as demonstrated in this study.

In the second set of experiments, we aimed to test the robustness of CDS and NaDS. For each network, we selected $10$ pseudo-random feasible points to use as starting points for both direct search approaches. These pseudo-random points were chosen to be similar to the SD solution: the starting points were constructed by sampling $B$ random nodes from the list of the top $4B$ nodes, where $B$ is the budget for the specific problem. This approach ensures that the starting points are random yet still possess a good value in terms of influence.
Additionally, we used a time-stopping condition for every problem, depending on the network size. For the larger networks, namely Arxiv Astro and Aemail-Enron, the time budget was $400$ seconds multiplied by the budget $B$, while for the other networks, it was $200$ seconds multiplied by the budget $B$.
In Figure~\ref{fig:avggap}, we report the average relative gap of the solutions of CDS and NaDS, where the relative gap at time $t$ for a specific problem is defined as:
\begin{equation}
g(t) = \frac{m - s(t)}{m},
\end{equation}
with $s(t)$ the solution found at time $t$ by the method and $m$ is the best solution found among all methods.

Table~\ref{tab:method_comparison} presents more detailed results from both sets of experiments and summarize the best outcomes achieved by classic models and heuristics (namely SD , SG, KC, CC, and CI) as well as the NaDS and CDS methods. For NaDS, we provide results from the first set of experiments, labeled as NaDS (disc), where the method was initialized using the solution from SD. Additionally, results from the second set of experiments, using 10 pseudo-random starting points, are included. For these, we use the labels NaDS (mean) for the average results across a single problem and NaDS (best) for the best result obtained.
For more comprehensive results from the second set of experiments, we refer to Section~\ref{sec:ext_res} of the Appendix (Section~\ref{sec:ext_res}).

Table 3 summarizes the observed time complexity across NaDS and CDS algorithms, revealing that NaDS maintains a better scaling behavior than CDS while giving much better results in terms of solution quality.  Notably, while greedy heuristics show good time complexity in theory, their actual cpu-time cost is much larger than our methods (and solution quality is not as good as NaDS). As for the other heuristics, while they achieve faster computation times, our experiments reveal this often comes at the cost of solution quality.

We can thus  draw the following conclusions from the numerical experience we carried out. More specifically, 
\begin{enumerate} [label=(\roman*)]
    \item Classic methods do not perform well in this setting. While the SD yields good results, it does not guarantee success across all datasets. Notably, in this setting, the submodularity of the objective functions does not hold, which was a crucial assumption for the guarantees provided by greedy methods. This implies that in this context, direct search methods have the upper hand, consistently showing better results. As we can see in Table~\ref{tab:method_comparison}, all the best results were achieved by direct search methods. In particular, NaDS achieved the best results for every problem except one.
   
    \item NaDS outperforms CDS in almost every dataset, both with SD and pseudo-random start. We also observe that the performance difference increases with the complexity of the problems, particularly represented by the budget $B$. This confirms that the improved neighborhood formulation of NaDS scales well with the dimensions of the graphs. We note that in every problem, except for one single case, NaDS achieves better results than CDS. Figure~\ref{fig:avggap} highlights the improved computational efficiency of the NaDS method.
\end{enumerate}

%and specifically for an in-depth view of each dataset
% \input{Tables/results}

% \begin{figure}[t]
%     \centering
%     \includegraphics[width=\columnwidth]{Images/avg_gap.png}
%     \caption{
%    Results of the second set of experiments where we compare the average relative gap over time 
%  (with shaded standard deviation) of the NaDS and CDS methods, resp. in blue and orange.}
%     \label{fig:avggap}
% \end{figure}

\section{Conclusions}

In this paper, we introduced  NaDS, a novel direct search method for solving the IM problem in networks. The NaDS method leverages the structure of the network to enhance computational efficiency and effectiveness in identifying influential nodes. Our approach integrates the GIP  model to better capture the dynamics of influence spread, providing a more robust and scalable solution compared to traditional direct search methods.

Extensive numerical experiments demonstrate that NaDS outperforms many classic methods, like, e.g., Single Discount  \cite{chen2010scalable}
        Simple Greedy (SG) \cite{kempe2003maximizing}, 
       Katz Centrality  \cite{tian2022unifying}, K-core Centrality \cite{Kitsak_2010} and 
       Collective Influence \cite{Morone_2015} 
in terms of 
% both computational time and 
quality of the solutions. In particular, 
in our GIP setting, some of those heuristics, like, e.g., greedy algorithms, lose their theoretical guarantees, especially those based on the submodularity of the objective function. This makes the use of direct search methods a valuable approach in practice.
 So, while those classic methods have proven successful in settings where the objective function exhibits submodularity \cite{mossel2010submodularity}, our work highlights their limitations in those complex scenarios, where such properties may not hold. This observation is consistent with recent critiques of heuristic methods in non-submodular settings \cite{arora2017debunking, tian2022unifying}.
 Furthermore, the superior performance of NaDS over the Customized Direct Search (CDS) method \cite{tian2022unifying} highlights the importance of incorporating network structure into the optimization process. By refining the neighborhood definition to focus on graph-connected nodes, NaDS indeed achieves significant gains in terms of computational efficiency and highly improved solution quality.  Our results thus complement recent efforts to unify propagation models \cite{tian2022unifying}, demonstrating that tailored direct search methods work pretty well when considering more sophisticated diffusion dynamics.

Future works will focus on refining the computational efficiency of the algorithmic framework, as the implementation still has many possible improvements. 
The current implementation of NaDS could indeed be optimized by, e.g., parallelizing the neighborhood evaluation phase, leveraging modern multi-core architectures. This will likely allow for further experimentation and application of the NaDS method to even larger networks, on the scale of hundreds of thousands or even millions of nodes.

Additionally, an important extension to the IM problem can be made by considering higher-order networks, i.e. hypergraphs and multilayer graphs.  In particular, extending NaDS to hypergraphs would require redefining the neighborhood to account for hyperedges, thus capturing the influence of node groups; while adapting it to multilayer networks would involve integrating cross-layer dependencies into the neighborhood formulation.

% By addressing all those challenges, NaDS could become a versatile tool for a broader range of IM applications while maintaining computational tractability.

%We are considering a novel problem formulation and a possible direct search approach based on NaDS.
\label{sec:concl}

% \section*{Acknowledgment}
% Insert the Acknowledgment text here.

\bibliographystyle{comnet}
\bibliography{sample-base}

\begin{thebibliography}{00}

\bibitem{arora2017debunking}
Arora, A., Galhotra, S. {\&} Ranu, S. (2017)  Debunking the myths of influence maximization: An in-depth benchmarking study. In {\em Proceedings of the 2017 ACM international conference on management of data}, pages 651--666.

\bibitem{AuHa2017}
Audet, C. {\&} Hare, W. (2017) {\em Derivative-Free and Blackbox Optimization}.
Springer Series in Operations Research and Financial Engineering. Springer, Cham, Switzerland.

\bibitem{bakshy2012role}
Bakshy, E., Rosenn, I., Marlow, C. {\&} Adamic, L. (2012)  The role of social networks in information diffusion. In {\em Proceedings of the 21st international conference on World Wide Web}, pages 519--528.

\bibitem{borgs2014maximizing}
Borgs, C., Brautbar, M., Chayes, J. {\&} Lucier, B. (2014)  Maximizing social influence in nearly optimal time. In {\em Proceedings of the twenty-fifth annual ACM-SIAM symposium on Discrete algorithms}, pages 946--957. SIAM.

\bibitem{bovet2019influence}
Bovet, A. {\&} Makse, H.~A. (2019)  Influence of fake news in Twitter during the 2016 US presidential election. {\em Nature communications}, \textbf{10}(1), 7.

\bibitem{budak2011limiting}
Budak, C., Agrawal, D. {\&} El~Abbadi, A. (2011)  Limiting the spread of misinformation in social networks. In {\em Proceedings of the 20th international conference on World Wide Web}, pages 665--674.

\bibitem{centola2010spread}
Centola, D. (2010)  The spread of behavior in an online social network experiment. {\em Science}, \textbf{329}(5996), 1194--1197.

\bibitem{chen2022information}
Chen, W., Castillo, C. {\&} Lakshmanan, L.~V. (2022) {\em Information and influence propagation in social networks}.
Springer Nature.

\bibitem{chen2012time}
Chen, W., Lu, W. {\&} Zhang, N. (2012)  Time-critical influence maximization in social networks with time-delayed diffusion process. In {\em Proceedings of the AAAI Conference on Artificial Intelligence}, volume~26, pages 591--598.

\bibitem{chen2010scalable}
Chen, W., Wang, C. {\&} Wang, Y. (2010)  Scalable influence maximization for prevalent viral marketing in large-scale social networks. In {\em Proceedings of the 16th ACM SIGKDD international conference on Knowledge discovery and data mining}, pages 1029--1038.

\bibitem{chen2009efficient}
Chen, W., Wang, Y. {\&} Yang, S. (2009)  Efficient influence maximization in social networks. In {\em Proceedings of the 15th ACM SIGKDD international conference on Knowledge discovery and data mining}, pages 199--208.

\bibitem{cheng2013staticgreedy}
Cheng, S., Shen, H., Huang, J., Zhang, G. {\&} Cheng, X. (2013)  Staticgreedy: solving the scalability-accuracy dilemma in influence maximization. In {\em Proceedings of the 22nd ACM international conference on information and knowledge management}, pages 509--518.

\bibitem{clifford1973model}
Clifford, P. {\&} Sudbury, A. (1973)  A model for spatial conflict. {\em Biometrika}, \textbf{60}(3), 581--588.

\bibitem{cohen2014sketch}
Cohen, E., Delling, D., Pajor, T. {\&} Werneck, R.~F. (2014)  Sketch-based influence maximization and computation: Scaling up with guarantees. In {\em Proceedings of the 23rd ACM international conference on conference on information and knowledge management}, pages 629--638.

\bibitem{conn2009introduction}
Conn, A.~R., Scheinberg, K. {\&} Vicente, L.~N. (2009) {\em Introduction to derivative-free optimization}.
SIAM.

\bibitem{domingos2001mining}
Domingos, P. {\&} Richardson, M. (2001)  Mining the network value of customers. In {\em Proceedings of the seventh ACM SIGKDD international conference on Knowledge discovery and data mining}, pages 57--66.

\bibitem{dzahini2024revisiting}
Dzahini, K., Rinaldi, F., Royer, C. {\&} Zeffiro, D. (2024)  Revisiting Theoretical Guarantees of Direct-Search Methods. {\em arXiv preprint arXiv:2403.05322}.

\bibitem{erkol2019systematic}
Erkol, S., Castellano, C. {\&} Radicchi, F. (2019)  Systematic comparison between methods for the detection of influential spreaders in complex networks. {\em Scientific Reports}, \textbf{9}(1).

\bibitem{goldenberg2001talk}
Goldenberg, J., Libai, B. {\&} Muller, E. (2001)  Talk of the network: A complex systems look at the underlying process of word-of-mouth. {\em Marketing letters}, \textbf{12}, 211--223.

\bibitem{rodriguez2011uncovering}
Gomez-Rodriguez, M., Balduzzi, D. {\&} Sch\"{o}lkopf, B. (2011)  Uncovering the temporal dynamics of diffusion networks. In {\em Proceedings of the 28th International Conference on International Conference on Machine Learning}, page 561–568, Madison, WI, USA. Omnipress.

\bibitem{granovetter1978threshold}
Granovetter, M. (1978)  Threshold models of collective behavior. {\em American journal of sociology}, \textbf{83}(6), 1420--1443.

\bibitem{guille2013information}
Guille, A., Hacid, H., Favre, C. {\&} Zighed, D.~A. (2013)  Information diffusion in online social networks: A survey. {\em ACM Sigmod Record}, \textbf{42}(2), 17--28.

\bibitem{SciPyProceedings_11}
Hagberg, A.~A., Schult, D.~A. {\&} Swart, P.~J. (2008)  Exploring Network Structure, Dynamics, and Function using NetworkX. In Varoquaux, G., Vaught, T. {\&} Millman, J., editors, {\em Proceedings of the 7th Python in Science Conference}, pages 11 -- 15, Pasadena, CA USA.

\bibitem{harris2020array}
Harris, C.~R., Millman, K.~J., Van Der~Walt, S.~J., Gommers, R., Virtanen, P., Cournapeau, D., Wieser, E., Taylor, J., Berg, S., Smith, N.~J.  et~al. (2020)  Array programming with NumPy. {\em Nature}, \textbf{585}(7825), 357--362.

\bibitem{he2012influence}
He, X., Song, G., Chen, W. {\&} Jiang, Q. (2012)  Influence blocking maximization in social networks under the competitive linear threshold model. In {\em Proceedings of the 2012 SIAM international conference on data mining}, pages 463--474. SIAM.

\bibitem{4160265}
Hunter, J.~D. (2007)  Matplotlib: A 2D Graphics Environment. {\em Computing in Science and Engineering}, \textbf{9}(3), 90--95.

\bibitem{jung2012irie}
Jung, K., Heo, W. {\&} Chen, W. (2012)  Irie: Scalable and robust influence maximization in social networks. In {\em 2012 IEEE 12th international conference on data mining}, pages 918--923. IEEE.

\bibitem{Katz_1953}
Katz, L. (1953)  A New Status Index Derived from Sociometric Analysis. {\em Psychometrika}, \textbf{18}(1), 39–43.

\bibitem{kempe2003maximizing}
Kempe, D., Kleinberg, J. {\&} Tardos, {\'E}. (2003)  Maximizing the spread of influence through a social network. In {\em Proceedings of the ninth ACM SIGKDD international conference on Knowledge discovery and data mining}, pages 137--146.

\bibitem{kermack1927contribution}
Kermack, W.~O. {\&} McKendrick, A.~G. (1927)  A contribution to the mathematical theory of epidemics. {\em Proceedings of the Royal Society London. Series A, Containing papers of a mathematical and physical character}, \textbf{115}(772), 700--721.

\bibitem{kim2014ct}
Kim, J., Lee, W. {\&} Yu, H. (2014)  CT-IC: Continuously activated and time-restricted independent cascade model for viral marketing. {\em Knowledge-Based Systems}, \textbf{62}, 57--68.

\bibitem{kimura2006tractable}
Kimura, M. {\&} Saito, K. (2006)  Tractable models for information diffusion in social networks. In {\em European conference on principles of data mining and knowledge discovery}, pages 259--271. Springer.

\bibitem{Kitsak_2010}
Kitsak, M., Gallos, L.~K., Havlin, S., Liljeros, F., Muchnik, L., Stanley, H.~E. {\&} Makse, H.~A. (2010)  Identification of influential spreaders in complex networks. {\em Nature Physics}, \textbf{6}(11), 888–893.

\bibitem{enron}
Klimt, B. {\&} Yang, Y. (2004)  The Enron Corpus: A New Dataset for Email Classification Research. In Boulicaut, J.-F., Esposito, F., Giannotti, F. {\&} Pedreschi, D., editors, {\em Machine Learning: ECML 2004}, pages 217--226, Berlin, Heidelberg. Springer Berlin Heidelberg.
\url{https://doi.org/10.1007/978-3-540-30115-8_22}.

\bibitem{Lancichinetti_2008}
Lancichinetti, A., Fortunato, S. {\&} Radicchi, F. (2008)  Benchmark graphs for testing community detection algorithms. {\em Physical Review E}, \textbf{78}(4).

\bibitem{larson2019derivative}
Larson, J., Menickelly, M. {\&} Wild, S.~M. (2019)  Derivative-free optimization methods. {\em Acta Numerica}, \textbf{28}, 287--404.

\bibitem{leskovec2007dynamics}
Leskovec, J., Adamic, L.~A. {\&} Huberman, B.~A. (2007a)  The dynamics of viral marketing. {\em ACM Transactions on the Web (TWEB)}, \textbf{1}(1), 5--es.

\bibitem{10.1145/1217299.1217301}
Leskovec, J., Kleinberg, J. {\&} Faloutsos, C. (2007b)  Graph evolution: Densification and shrinking diameters. {\em ACM Transactions on Knowledge Discovery from Data}, \textbf{1}(1), 2–es.

\bibitem{leskovec2007cost}
Leskovec, J., Krause, A., Guestrin, C., Faloutsos, C., VanBriesen, J. {\&} Glance, N. (2007c)  Cost-effective outbreak detection in networks. In {\em Proceedings of the 13th ACM SIGKDD international conference on Knowledge discovery and data mining}, pages 420--429.

\bibitem{snapnets}
Leskovec, J. {\&} Krevl, A. (2014)  {SNAP Datasets}: {Stanford} Large Network Dataset Collection. \url{http://snap.stanford.edu/data}.

\bibitem{leskovec2008communitystructurelargenetworks}
Leskovec, J., Lang, K.~J., Dasgupta, A. {\&} Mahoney, M.~W. (2008)  Community Structure in Large Networks: Natural Cluster Sizes and the Absence of Large Well-Defined Clusters. .

\bibitem{im_survey}
Li, Y., Fan, J., Wang, Y. {\&} Tan, K.-L. (2018)  Influence Maximization on Social Graphs: A Survey. {\em IEEE Transactions on Knowledge and Data Engineering}, \textbf{30}(10), 1852--1872.

\bibitem{liu2012time}
Liu, B., Cong, G., Xu, D. {\&} Zeng, Y. (2012)  Time constrained influence maximization in social networks. In {\em 2012 IEEE 12th international conference on data mining}, pages 439--448. IEEE.

\bibitem{liu2014influence}
Liu, Q., Xiang, B., Chen, E., Xiong, H., Tang, F. {\&} Yu, J.~X. (2014)  Influence maximization over large-scale social networks: A bounded linear approach. In {\em Proceedings of the 23rd ACM international conference on information and knowledge management}, pages 171--180.

\bibitem{10.5555/2999134.2999195}
McAuley, J. {\&} Leskovec, J. (2012)  Learning to discover social circles in ego networks. In {\em Proceedings of the 26th International Conference on Neural Information Processing Systems - Volume 1}, NIPS'12, page 539–547, Red Hook, NY, USA. Curran Associates Inc.

\bibitem{Morone_2015}
Morone, F. {\&} Makse, H.~A. (2015)  Influence maximization in complex networks through optimal percolation. {\em Nature}, \textbf{524}(7563), 65–68.

\bibitem{mossel2010submodularity}
Mossel, E. {\&} Roch, S. (2010)  Submodularity of influence in social networks: From local to global. {\em SIAM Journal on Computing}, \textbf{39}(6), 2176--2188.

\bibitem{nekovee2007theory}
Nekovee, M., Moreno, Y., Bianconi, G. {\&} Marsili, M. (2007)  Theory of rumour spreading in complex social networks. {\em Physica A: Statistical Mechanics and its Applications}, \textbf{374}(1), 457--470.

\bibitem{ohsaka2014fast}
Ohsaka, N., Akiba, T., Yoshida, Y. {\&} Kawarabayashi, K.-i. (2014)  Fast and accurate influence maximization on large networks with pruned Monte-Carlo simulations. In {\em Proceedings of the AAAI Conference on Artificial Intelligence}, volume~28.

\bibitem{pastor2015epidemic}
Pastor-Satorras, R., Castellano, C., Van~Mieghem, P. {\&} Vespignani, A. (2015)  Epidemic processes in complex networks. {\em Reviews of modern physics}, \textbf{87}(3), 925.

\bibitem{divideandconquer}
Patwardhan, S., Radicchi, F. {\&} Fortunato, S. (2023)  Influence maximization: Divide and conquer. {\em Physical Review E}, \textbf{107}, 054306.

\bibitem{10.1093/comnet/cnz029}
Pei, S., Wang, J., Morone, F. {\&} Makse, H.~A. (2019)  Influencer identification in dynamical complex systems. {\em Journal of Complex Networks}, \textbf{8}(2), cnz029.

\bibitem{feather}
Rozemberczki, B. {\&} Sarkar, R. (2020)  {Characteristic Functions on Graphs: Birds of a Feather, from Statistical Descriptors to Parametric Models}. In {\em Proceedings of the 29th ACM international conference on information and knowledge management}, page 1325–1334. ACM.

\bibitem{schelling2006micromotives}
Schelling, T.~C. (2006) {\em Micromotives and macrobehavior}.
WW Norton \& Company.

\bibitem{shakarian2015independent}
Shakarian, P., Bhatnagar, A., Aleali, A., Shaabani, E. {\&} Guo, R. (2015)  The independent cascade and linear threshold models. {\em Diffusion in Social Networks}, pages 35--48.

\bibitem{sun2011survey}
Sun, J. {\&} Tang, J. (2011)  A survey of models and algorithms for social influence analysis. {\em Social network data analytics}, pages 177--214.

\bibitem{tejaswi2016diffusion}
Tejaswi, V., Bindu, P. {\&} Thilagam, P.~S. (2016)  Diffusion models and approaches for influence maximization in social networks. In {\em 2016 international conference on advances in computing, communications and informatics (ICACCI)}, pages 1345--1351. IEEE.

\bibitem{tian2022unifying}
Tian, Y. {\&} Lambiotte, R. (2022)  Unifying information propagation models on networks and influence maximization. {\em Physical Review E}, \textbf{106}(3), 034316.

\bibitem{virtanen2020scipy}
Virtanen, P., Gommers, R., Oliphant, T.~E., Haberland, M., Reddy, T., Cournapeau, D., Burovski, E., Peterson, P., Weckesser, W., Bright, J.  et~al. (2020)  SciPy 1.0: fundamental algorithms for scientific computing in Python. {\em Nature methods}, \textbf{17}(3), 261--272.

\bibitem{xie2015dynadiffuse}
Xie, M., Yang, Q., Wang, Q., Cong, G. {\&} Melo, G. (2015)  Dynadiffuse: A dynamic diffusion model for continuous time constrained influence maximization. In {\em Proceedings of the AAAI Conference on Artificial Intelligence}, volume~29, page 346 – 352.

\bibitem{ye2012exploring}
Ye, M., Liu, X. {\&} Lee, W.-C. (2012)  Exploring social influence for recommendation: a generative model approach. In {\em Proceedings of the 35th international ACM SIGIR conference on Research and development in information retrieval}, pages 671--680.

\bibitem{zhang2014recent}
Zhang, H., Mishra, S., Thai, M.~T., Wu, J. {\&} Wang, Y. (2014)  Recent advances in information diffusion and influence maximization in complex social networks. {\em Opportunistic Mobile Social Networks}, \textbf{37}(1.1), 37.

\end{thebibliography}

%Any appendices should precede the Reference section.
\appendix

\section{Tests on artificial graphs}\label{sec:artificial}
In this section, we present a series of tests on artificial graphs. These graphs exhibit diverse structural characteristics, since we aimed to cover a broad range of graph types. In particular, we conducted experiments with Lancichinetti-Fortunato-Radicchi (LFR) graphs~\cite{Lancichinetti_2008}: We varied the number of nodes $n \in \{1000, 2000, 5000\}$ and the mixing parameter $\mu \in \{0.1, 0.5\}$, which controls the fraction of edges a node shares with nodes in other communities. This resulted in 6 different graph instances.
% \end{itemize}}}
% \begin{itemize}
%     \item \textbf{Barabási–Albert graphs}~\cite{barabasi-albert}: We varied both the total number of nodes $n \in \{100, 150, 200\}$ and the number of edges $m \in \{1, 2, 4, 10\}$ attached from a new node to existing nodes, resulting in a total of 12 different graphs.   
%     \item \textbf{Stochastic Block Model (SBM) graphs}~\cite{sbm}: We varied the number and size of communities, as well as the intra-community probability $p_1 \in \{0.3, 0.5\}$ and inter-community probability $p_2 \in \{0.05, 0.2\}$. Community size configurations included $[100, 100]$, $[150, 50]$, $[100, 50, 50]$, and $[50, 50, 50, 50]$, yielding a total of 16 distinct graphs.
%     \item \textbf{Lancichinetti–Fortunato–Radicchi (LFR) graphs}~\cite{Lancichinetti_2008}: We varied the number of nodes $n \in \{100, 150, 200\}$ and the mixing parameter $\mu \in \{0.1, 0.3, 0.5, 0.7\}$, which controls the fraction of edges a node shares with nodes in other communities. This resulted in 12 different graph instances.
% \end{itemize}}

The experiments were carried out using the same parameter settings as those used in the real-world graph experiments (see Section~\ref{sec:results}). Specifically, we set $\theta_l = 2$ and $\theta_h = 50$. For the budget $B$, we used the values $5$, $10$, $15$, and $20$.
We adopted the same methodology as in the previous set of experiments: we compared the results obtained by the classic methods and heuristics \textbf{SD}, \textbf{SG}, \textbf{KC}, \textbf{CC}, and \textbf{CI}. Both \textbf{CDS} and \textbf{NaDS} were initialized using the \textbf{SD} solution, which once again proved to be the most effective starting point.
For both NaDS and CDS, we allocated a time budget equal to $20 \times B$ seconds. In Table~\ref{tab:lfr}, we report extensive results obtained on the LFR benchmark graphs. As shown, direct search methods consistently outperform the alternative heuristics, which generally struggle to achieve comparable solution quality. Focusing on the comparison between NaDS and CDS, we observe that NaDS achieves the best score in 22 out of 24 experiments. Notably, the performance gap between the two algorithms widens as the network size increases, highlighting the superior scalability of NaDS.

\begin{table}[t]
 
\caption{Performance comparison (in terms of cumulative influence score and computational time in seconds) based on cumulative influence on the Lancichinetti–Fortunato–Radicchi (LFR), evaluated across different network parameters (number of nodes and mixing parameter $\mu$) and budget levels. The best results are highlighted in bold.}
\centering
\scriptsize
\setlength{\tabcolsep}{1.2pt}
\begin{tabular}{lllccccccccccccccc} 
\toprule
\multirow{2}{*}{\textbf{\#nodes}} & \multirow{2}{*}{$\boldsymbol{\mu}$} & \multirow{2}{*}{\textbf{B}} 
& \multicolumn{2}{c}{\textbf{NaDS}} 
& \multicolumn{2}{c}{\textbf{CDS}} 
& \multicolumn{2}{c}{\textbf{SD}} 
& \multicolumn{2}{c}{\textbf{SG}} 
& \multicolumn{2}{c}{\textbf{KC}} 
& \multicolumn{2}{c}{\textbf{CC}} 
& \multicolumn{2}{c}{\textbf{CI}} \\
\cmidrule(lr){4-5} \cmidrule(lr){6-7} \cmidrule(lr){8-9} \cmidrule(lr){10-11} \cmidrule(lr){12-13} \cmidrule(lr){14-15} \cmidrule(lr){16-17}
 & & & Score & Time & Score & Time & Score & Time & Score & Time & Score & Time & Score & Time & Score & Time \\
\midrule
\multirow{8}{*}{1000} & \multirow{4}{*}{0.1} 
& 5  & \textbf{131.496}   & 45.924   & \textbf{131.496}    & 44.159    & \textbf{131.496}    & 0.003           & 10.138       & 7.920       & 26.454      & 0.062      & 35.359       & 0.016       & 104.935   & 2.360    \\
 &  & 10 & \textbf{204.773}   & 195.358  & \textbf{204.773}    & 179.544   & 192.804     & 0.016           & 67.838       & 20.226      & 44.558      & 0.064      & 96.074       & 0.016       & 183.946   & 3.783    \\
 &  & 15 & 255.156   & 160.258  & \textbf{257.809}   & 300   & 225.994     & 0.016           & 102.547      & 33.360      & 56.623      & 0.058      & 144.525      & 0.016       & 229.348   & 5.299    \\
 &  & 20 & \textbf{312.928}   & 400 & \textbf{312.928}   & 400  & 299.813     & 0.016           & 163.165      & 46.064      & 79.628      & 0.053      & 176.542      & 0.016       & 299.666   & 6.641    \\
 \cmidrule{2-17}
 & \multirow{4}{*}{0.5} & 5  & \textbf{116.847}   & 50.453   & \textbf{116.847}   & 100   & 73.639      & 0.016           & 54.469       & 8.153       & 45.181      & 0.062      & 0.8          & 0.016       & 89.099    & 4.662    \\
 &  & 10 & \textbf{209.474}   & 159.785  & \textbf{209.474}    & 185.104   & 199.115     & 0.001 & 183.232      & 21.978      & 105.176     & 0.062      & 16.229       & 0.016       & 194.428   & 7.918    \\
 &  & 15 & 284.536   & 300  & \textbf{285.706}    & 300   & 273.187     & 0.016           & 263.69       & 35.487      & 154.876     & 0.078      & 43.807       & 0.016       & 276.064   & 11.029   \\
 &  & 20 & \textbf{344.337}   & 400  & 343.017    & 400   & 330.312     & 0.016           & 336.25       & 48.521      & 202.412     & 0.076      & 62.602       & 0.018       & 328.96    & 13.770   \\
 \midrule
\multirow{8}{*}{2000} & \multirow{4}{*}{0.1} & 5  & 260.675   & 100  & \textbf{267.015}    & 100   & 210.693     & 0.047           & 2.704        & 14.869      & 32.749      & 0.352      & 132.681      & 0.047       & 210.693   & 8.317    \\
 &  & 10 & \textbf{412.057}   & 200  & 403.444    & 200   & 393.181     & 0.047           & 146.831      & 52.185      & 73.679      & 0.268      & 235.524      & 0.032       & 393.181   & 12.731   \\
 &  & 15 & \textbf{498.905}   & 300  & 469.297    & 300   & 448.701     & 0.047           & 386.527      & 92.955      & 97.86       & 0.508      & 337.595      & 0.031       & 465.459   & 17.616   \\
 &  & 20 & \textbf{563.99}    & 400  & 537.594    & 400   & 519.091     & 0.047           & 465.08       & 132.121     & 114.537     & 0.265      & 397.718      & 0.047       & 528.72    & 21.978   \\
 \cmidrule{2-17}
 & \multirow{4}{*}{0.5} & 5  & \textbf{174.515}   & 100  & 156.609    & 100   & 140.72      & 0.033           & 62.647       & 20.632      & 38.161      & 0.334      & 0.2          & 0.031       & 190.467   & 14.711   \\
 &  & 10 & \textbf{348.266}   & 200  & \textbf{348.266}    & 200   & \textbf{348.266}     & 0.048           & 293.524      & 62.632      & 95.963      & 0.298      & 27.18        & 0.042       & 350.832   & 25.253   \\
 &  & 15 & \textbf{486.225}   & 300  & \textbf{486.225}    & 300   & 480.068     & 0.047           & 431.186      & 102.975     & 169.974     & 0.293      & 68.305       & 0.047       & 480.068   & 34.630   \\
 &  & 20 & \textbf{565.554}   & 400  & 562.094    & 400   & 556.632     & 0.047           & 545.575      & 144.193     & 221.997     & 0.342      & 83.782       & 0.047       & 562.081   & 43.689   \\
\midrule
\multirow{8}{*}{5000} & \multirow{4}{*}{0.1} & 5  & \textbf{428.389}   & 100  & 320.864    & 100   & 290.276     & 0.336           & 10.766       & 129.576     & 36.976      & 2.585      & 66.57        & 0.110       & 290.276   & 29.457   \\
 &  & 10 & \textbf{568.159}   & 200  & 518.015    & 200  & 518.014     & 0.271           & 118.891      & 328.337     & 62.674      & 2.612      & 160.628      & 0.094       & 534.707   & 50.408   \\
 &  & 15 & \textbf{729.797}   & 300  & 683.534    & 300   & 665.404     & 0.268           & 511.435      & 572.642     & 84.292      & 2.623      & 287.186      & 0.109       & 665.404   & 71.135   \\
 &  & 20 & \textbf{841.139}   & 400  & 816.961    & 400  & 816.961     & 0.271           & 652.237      & 824.259     & 100.237     & 2.814      & 366.987      & 0.111       & 819.507   & 89.843   \\
 \cmidrule{2-17}
 & \multirow{4}{*}{0.5} & 5  & \textbf{358.507}   & 100 & 353.704    & 100  & 340.548     & 0.307           & 232.982      & 118.966     & 55.829      & 2.517      & 0.0          & 0.125       & 340.548   & 59.286   \\
 &  & 10 & \textbf{545.81}    & 200 & \textbf{545.81}     & 200  & 536.132     & 0.262           & 488.025      & 375.165     & 131.472     & 2.552      & 26.888       & 0.126       & 536.132   & 103.314  \\
 &  & 15 & \textbf{719.843}   & 300  & 694.727    & 300   & 687.606     & 0.257           & 661.534      & 626.514     & 172.1       & 2.948      & 32.512       & 0.125       & 687.606   & 141.707  \\
 &  & 20 & \textbf{862.267}   & 400  & 841.167    & 400   & 841.167     & 0.261           & 786.865      & 886.756     & 223.931     & 2.583      & 154.249      & 0.101       & 841.167   & 180.534 \\
\bottomrule
\end{tabular}
\label{tab:lfr}
\end{table}

\section{Additional Results}\label{sec:ext_res}
In this section we report comprehensive results from the second set of experiments presented in Section~\ref{sec:results}.
%The results include data from both sets of experiments and are organized as follows:
Figures~\ref{fig:arxivastro} to \ref{fig:facebook} focus on the comparison between NaDS and CDS. We plot the best result found by both methods over time for each value of budget $B$ and different networks, averaged over the 10 different pseudo-random starts, to obtain more stable results. We also show, as a reference, the objective value achieved by SD, since we consider it the best competitor to direct search methods.

\clearpage
%\begin{figure}[h]
\begin{figure*}[!t]
    \centering
\includegraphics[scale=0.26]{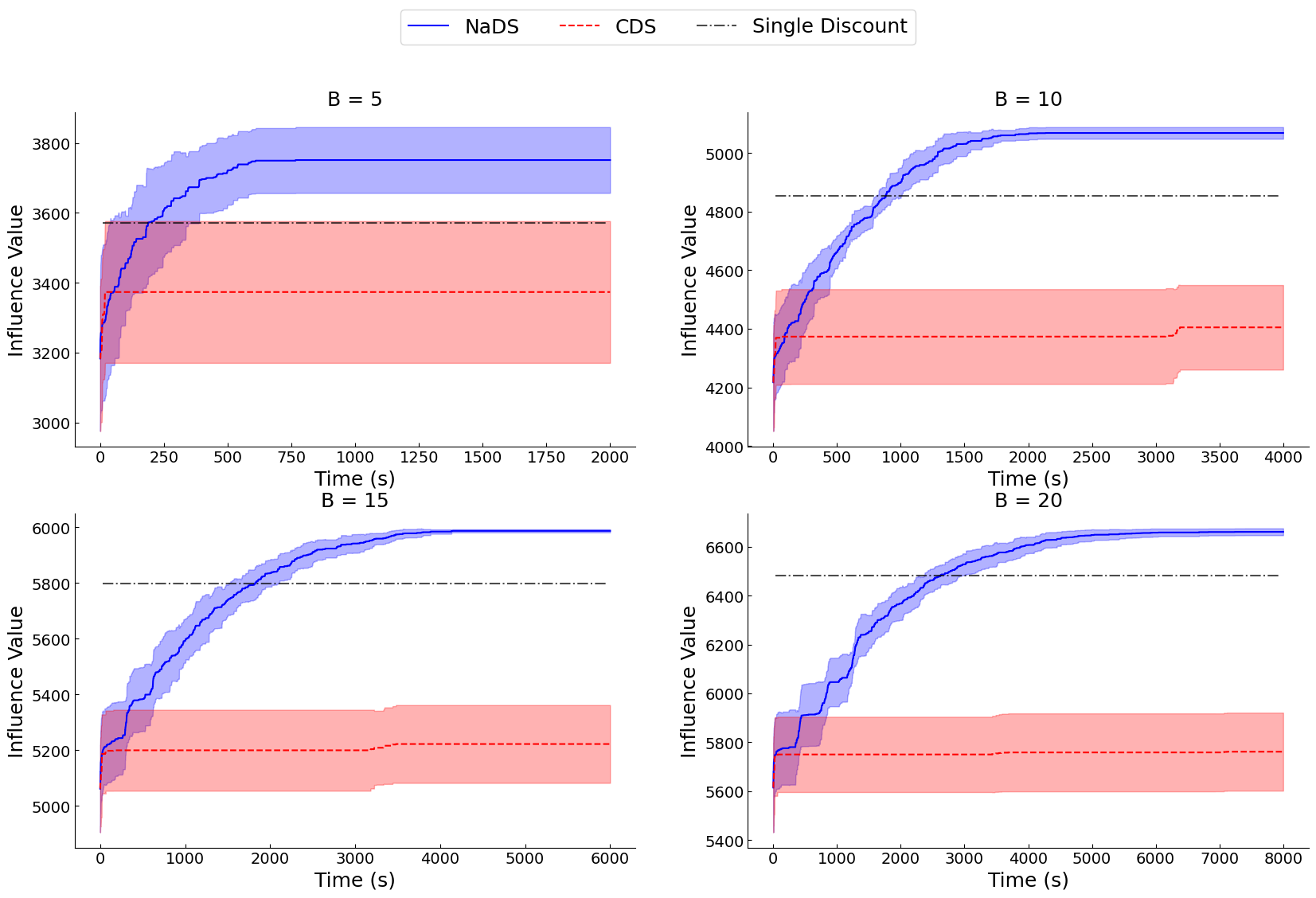}
    \caption{   
 Average relative gaps over time (with shaded standard deviation) across $10$ pseudo-random feasible points of the second set of experiments on the Arxiv Astro dataset.The NaDS method is represented by the blue solid line, the CDS method by the red dashed line, and the Single Discount solution by the black dash-dotted line.}  Each panel corresponds to a different budget value $\mathbf{B = 5, 10, 15, 20}$.
%Average results of the second set of experiments on the Arxiv Astro dataset, where we compare the average relative gap over time of the NaDS, CDS methods (resp. in blue and orange solid line), and the Single Discount solution (in dashed red line).  Each panel corresponds to a different budget value $\mathbf{K = 5, 10, 15, 20}$.
% The NaDS and CDS methods are represented by the blue and red solid lines, respectively, while the Single Discount solution is shown with a dashed black line
    \label{fig:arxivastro}
\end{figure*}

\begin{figure*}[!t]
    \centering
    \includegraphics[scale=0.26]{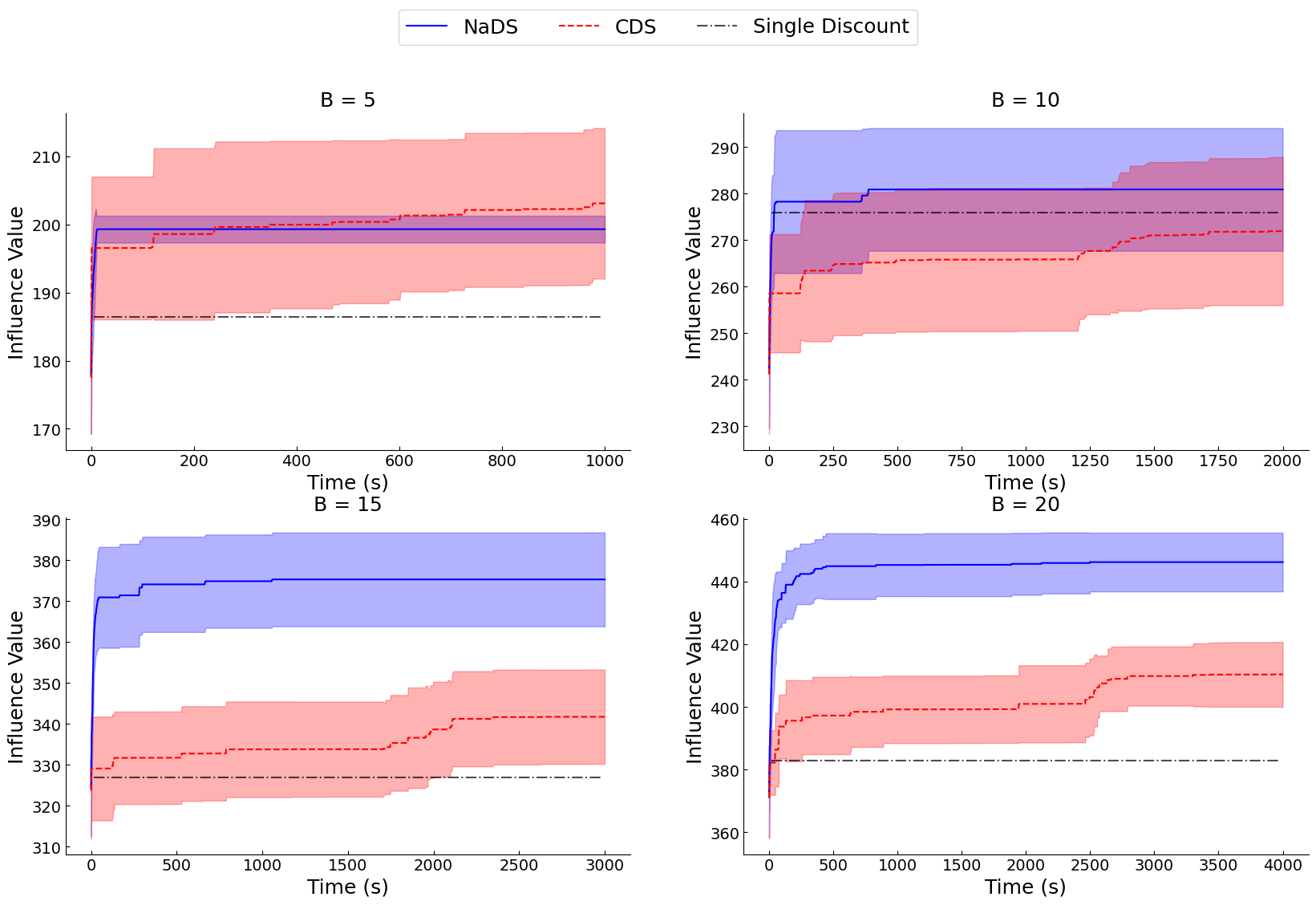}
    \caption{
 Average relative gaps over time (with shaded standard deviation) across $10$ pseudo-random feasible points of the second set of experiments on the Arxiv Gr-Qc dataset. The NaDS method is represented by the blue solid line, the CDS method by the red dashed line, and the Single Discount solution by the black dash-dotted line.  Each panel corresponds to a different budget value $\mathbf{B = 5, 10, 15, 20}$.}
    \label{fig:arxivgrqc}
\end{figure*}

\begin{figure*}[!t]
    \centering
    \includegraphics[scale=0.26]{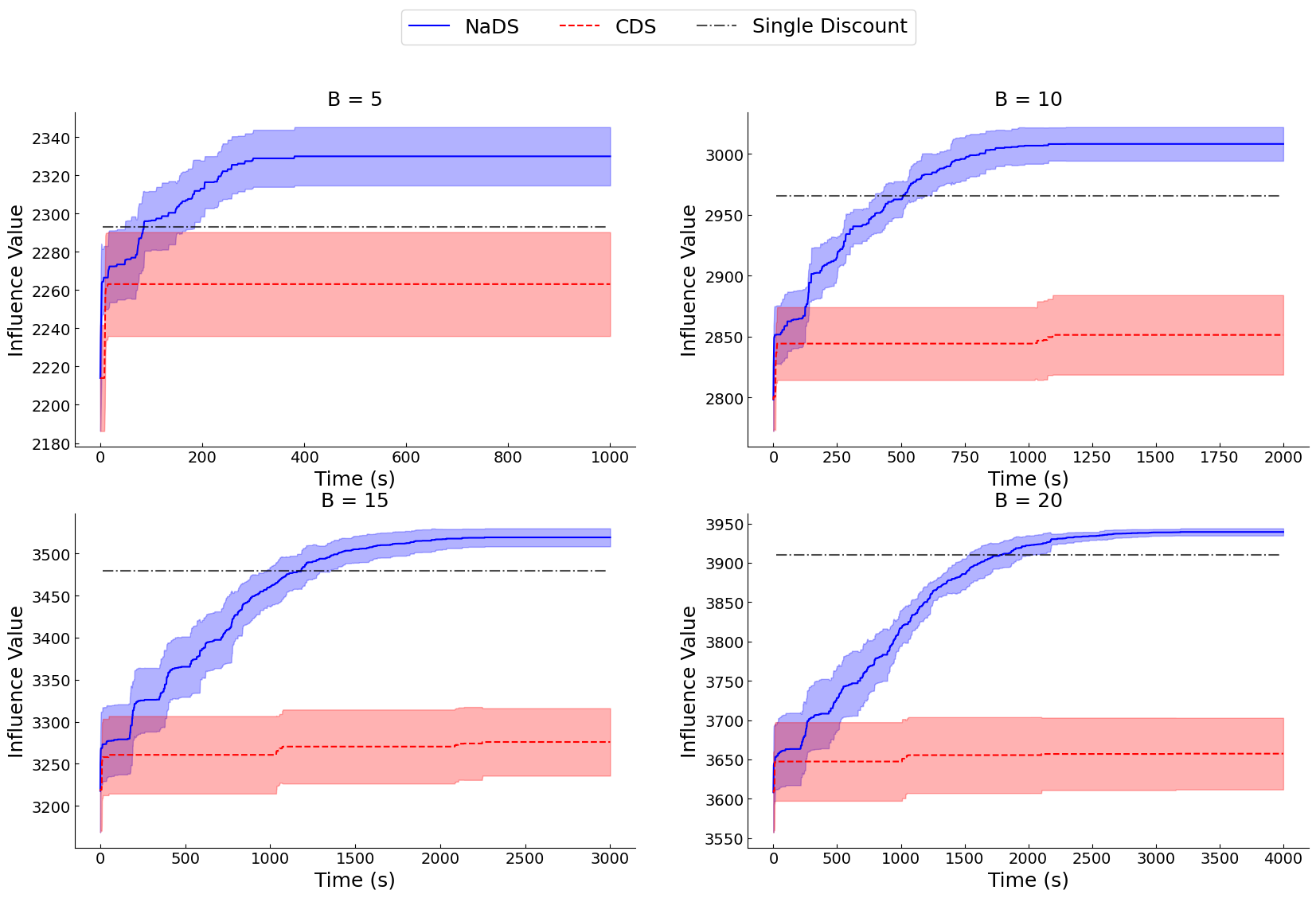}
    \caption{
 Average relative gaps over time (with shaded standard deviation) across $10$ pseudo-random feasible points of the second set of experiments on the Arxiv HeP-Ph dataset. The NaDS method is represented by the blue solid line, the CDS method by the red dashed line, and the Single Discount solution by the black dash-dotted line.  Each panel corresponds to a different budget value $\mathbf{B = 5, 10, 15, 20}$.}
    \label{fig:arxivhepph}
\end{figure*}

\begin{figure*}[!t]
    \centering
    \includegraphics[scale=0.26]{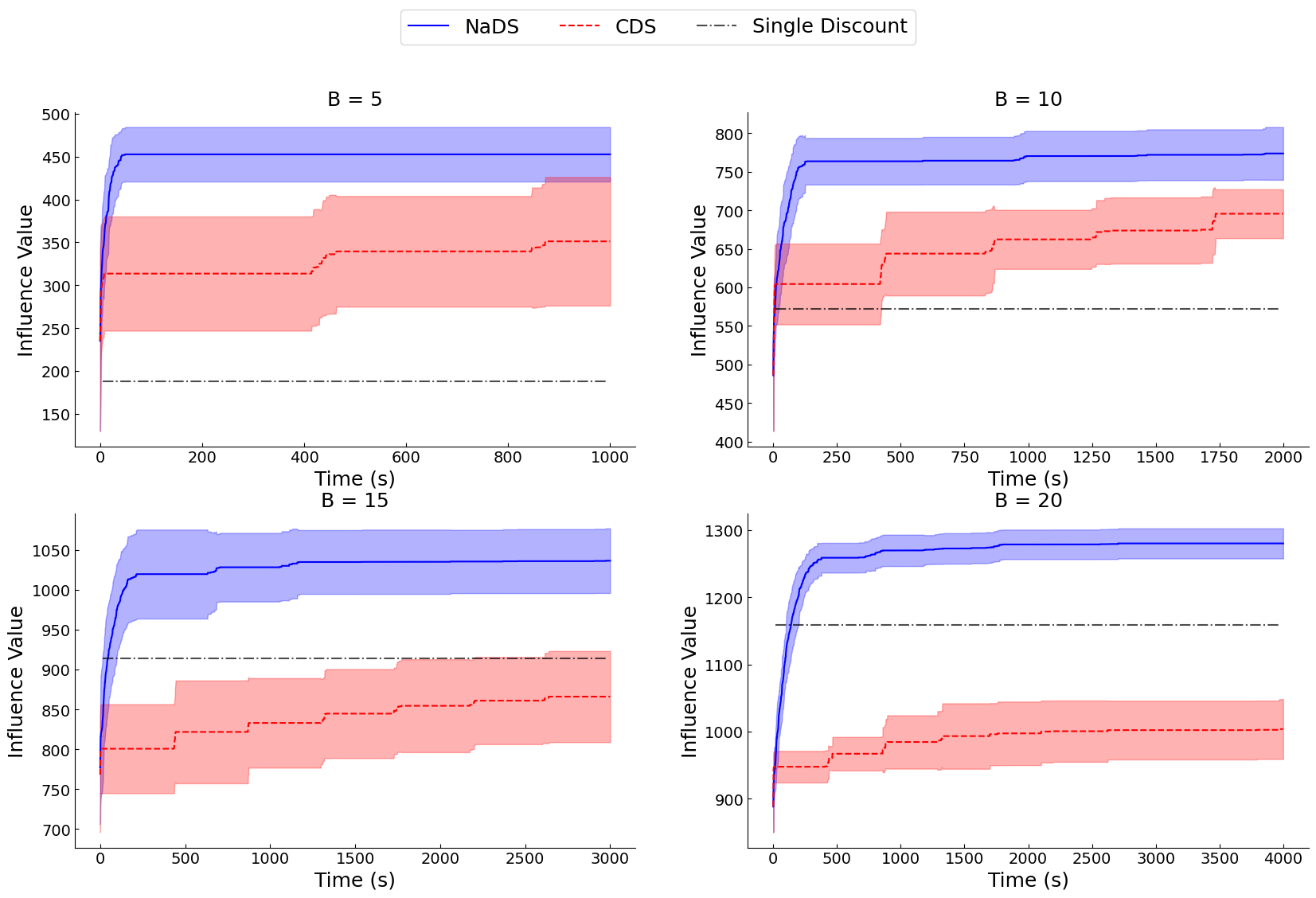}
    \caption{        
     Average relative gaps over time (with shaded standard deviation) across $10$ pseudo-random feasible points of the second set of experiments on the Lastfm-Asia dataset. The NaDS method is represented by the blue solid line, the CDS method by the red dashed line, and the Single Discount solution by the black dash-dotted line.  Each panel corresponds to a different budget value $\mathbf{B = 5, 10, 15, 20}$.}
    \label{fig:lastfmasia}
\end{figure*}

\begin{figure*}[!t]
    \centering
    \includegraphics[scale=0.26]{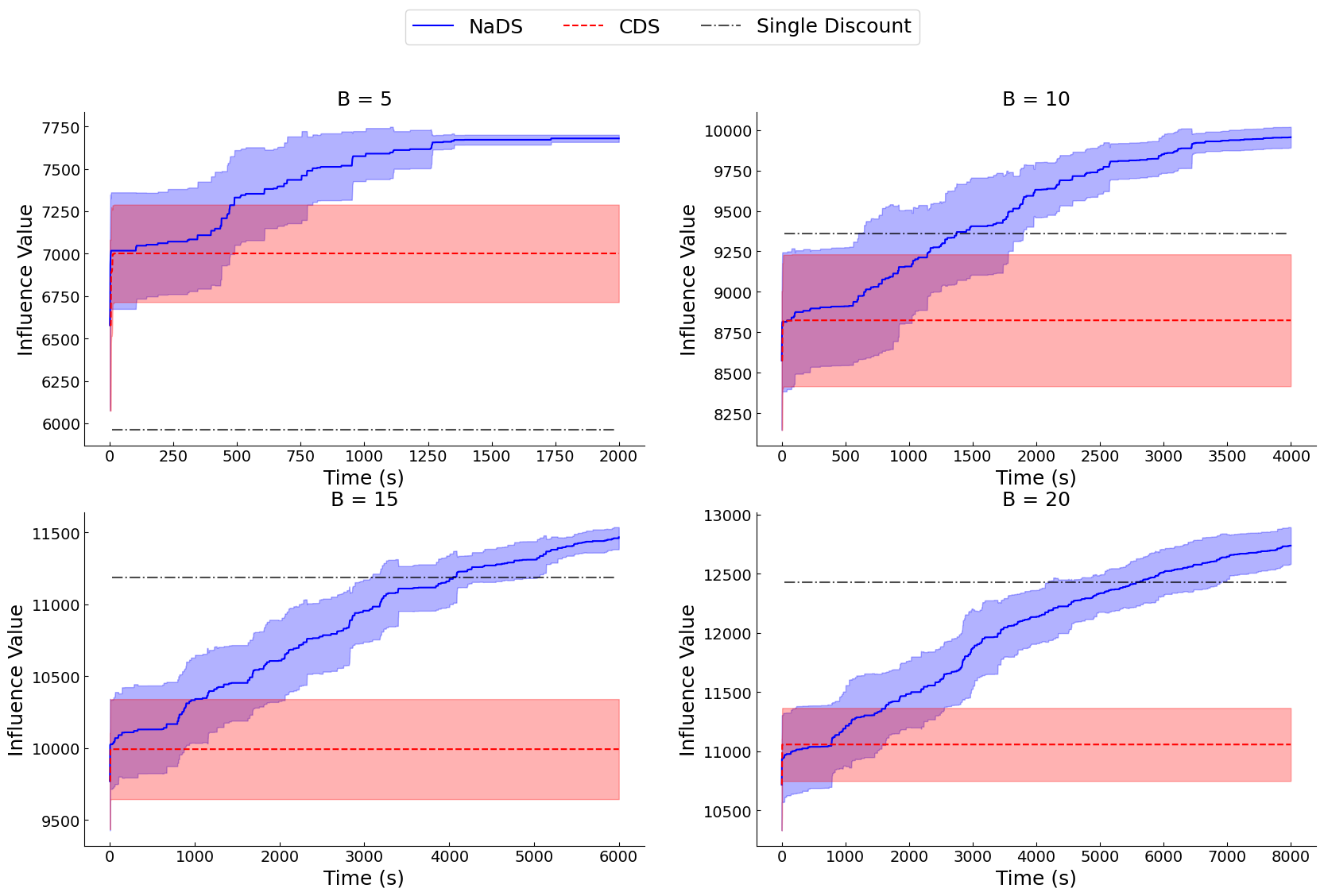}
    \caption{        
     Average relative gaps over time (with shaded standard deviation) across $10$ pseudo-random feasible points of the second set of experiments on the email-Enron dataset. The NaDS method is represented by the blue solid line, the CDS method by the red dashed line, and the Single Discount solution by the black dash-dotted line.  Each panel corresponds to a different budget value $\mathbf{B = 5, 10, 15, 20}$.}
    \label{fig:emailenron}
\end{figure*}

\begin{figure*}[!t]
    \centering
    \includegraphics[scale=0.26]{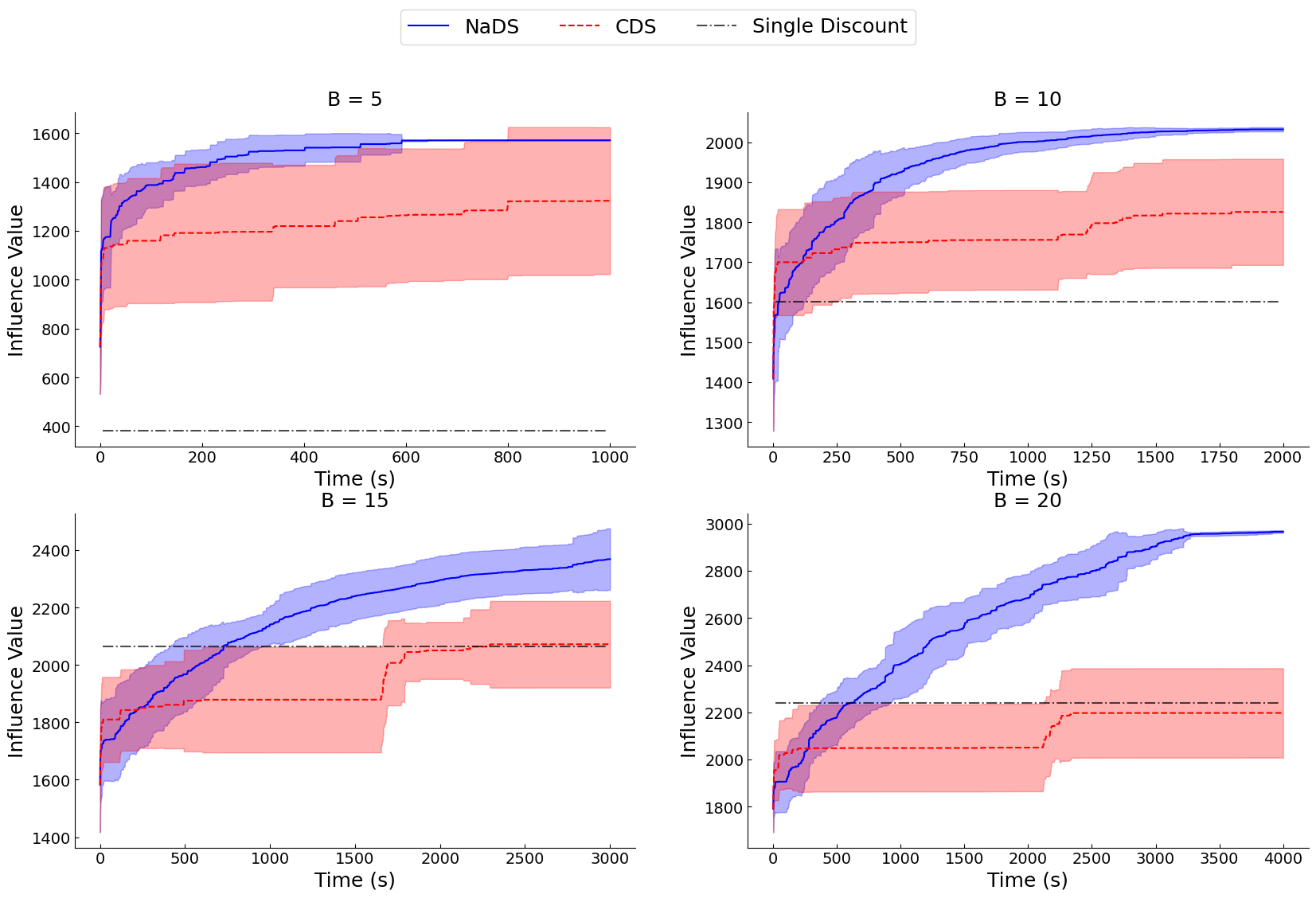}
    \caption{        Average relative gaps over time (with shaded standard deviation) across $10$ pseudo-random feasible points of the second set of experiments on the Facebook dataset. The NaDS method is represented by the blue solid line, the CDS method by the red dashed line, and the Single Discount solution by the black dash-dotted line.  Each panel corresponds to a different budget value $\mathbf{B = 5, 10, 15, 20}$.}
    \label{fig:facebook}
\end{figure*}

\end{document}